\documentclass[a4paper,USenglish,cleveref, autoref, thm-restate]{lipics-v2021}

\usepackage{amssymb}
\usepackage{amsmath}
\usepackage{comment}
\usepackage[color=gray!27]{todonotes}
\usepackage[utf8]{inputenc}
\usepackage{shuffle}
\usepackage{multirow}
\usepackage{mathabx}
\usepackage{mathrsfs}
\usepackage{wasysym}

\usepackage{apxproof}
\usepackage{xspace}

\newtheoremrep{theorem}{Theorem}[section]
\newtheoremrep{proposition}[theorem]{Proposition}
\newtheoremrep{lemma}[theorem]{Lemma}
\newtheoremrep{claim}[theorem]{Claim}
\newtheoremrep{conjecture}[theorem]{Conjecture}
\newtheoremrep{corollary}[theorem]{Corollary}
\theoremstyle{definition}
\newtheoremrep{definition}[theorem]{Definition}
\newtheoremrep{example}[theorem]{Example}
\theoremstyle{definition}
\newtheoremrep{remark}[theorem]{Remark}

\usepackage{hyperref}

\usepackage{tikz}
\usetikzlibrary{arrows,shapes,automata,positioning}
\usetikzlibrary{arrows.meta,bending,chains,matrix,positioning,scopes,calc,automata, shapes, patterns}
\tikzstyle{transition} = [font=\small]
\tikzstyle{line} = [draw,thick,-latex]
\usetikzlibrary{fadings,patterns}
\usetikzlibrary{intersections}
\usetikzlibrary{decorations.pathreplacing,calligraphy}
\pgfdeclarelayer{background}
\pgfsetlayers{background,main}
\usetikzlibrary{positioning}

\newcommand{\NL}{\textsf{NL}\xspace}
\newcommand{\NP}{\textsf{NP}\xspace}
\newcommand{\PSPACE}{\textsf{PSPACE}\xspace}

\newcommand{\PTIME}{\textsf{P}\xspace}
\newcommand{\XP}{\textsf{XP}\xspace}
\newcommand{\paraNP}{\textsf{para-NP}\xspace}
\newcommand{\W}{\textsf{W}}
\newcommand{\A}{\textsf{A}}
\newcommand{\co}{\textsf{co-}}
\newcommand{\FPT}{\textsf{FPT}\xspace}
\newcommand{\XNL}{\textsf{XNL}\xspace}
\newcommand{\WNL}{\textsf{WNL}\xspace}
\newcommand{\sync}{\textsf{sync}}
\newcommand{\WP}{\W[\textsf{P}]}
\newcommand{\WSAT}{\W[\textsf{SAT}]}
\newcommand{\fpt}{\textrm{fpt}}

\newcommand{\yes}{Yes}

\newcommand{\DFAINE}{\textsc{DFA-Intersection-Nonemptiness}\xspace}

\newcommand{\DFAINEshort}{\text{\textsc{DFA}-$\,\bigcap\neq\emptyset$}\xspace}
\newcommand{\BDFAINEshort}{\text{\textsc{BDFA}-$\,\bigcap\neq\emptyset$}\xspace}
\newcommand{\NFAINEshort}{\text{\textsc{NFA}-$\,\bigcap\neq\emptyset$}\xspace}
\newcommand{\BNFAINEshort}{\text{\textsc{BNFA}-$\,\bigcap\neq\emptyset$}\xspace}

\newcommand{\TNEJ}{\textsc{Tables-Non-Empty-Join}\xspace}
\newcommand{\TNEJshort}{\textsc{TNEJ}\xspace}

\newcommand{\comm}{\text{comm-}}
\newcommand{\unary}{\text{uny-}}
\newcommand{\bounded}{\text{bd-}}
\newcommand{\sbounded}{\text{s-bd-}}
\newcommand{\sparse}{\text{sparse-}}

\newcommand{\binBowtie}{\mathbin{\Bowtie}}

\newcommand{\dom}{\textit{dom}}
\newcommand{\range}{\textit{range}}

\newcommand{\Oh}{\mathcal{O}}

\DeclareMathOperator{\lcm}{lcm}
\usepackage{tcolorbox}
\newcommand{\problemdef}[4]{
	\begin{tcolorbox}[width = \textwidth,colback=white,arc=0pt,outer arc=0pt,boxrule=0.7pt,left =0.5em,right=0em]
		\textsc{#1} #2
		\\[2pt]
		\begin{tabular}{ @{}l p{0.84\textwidth} c }
			\textsf{Input:} & #3 \\[.5pt]
			\textsf{Problem:} & #4
		\end{tabular}
	\vspace{-0.25em}
	\end{tcolorbox}
}
\title{Finite Automata Intersection Non-Emptiness:  Parameterized Complexity Revisited}
\author{Henning Fernau}
{Informatikwissenschaften, FB IV, Universit\"at Trier, Germany}
{fernau@informatik.uni-trier.de}
{https://orcid.org/0000-0002-4444-3220}
{}
\author{Stefan Hoffmann}
{Informatikwissenschaften, FB IV, Universit\"at Trier, Germany}
{hoffmanns@informatik.uni-trier.de}
{https://orcid.org/0000-0002-7866-075X}
{}
\author{Michael Wehar}
{Swarthmore College, Swarthmore, PA, USA}
{mwehar1@swarthmore.edu}
{}
{}

\authorrunning{H. Fernau, S. Hoffmann and M. Wehar} 

\Copyright{Henning Fernau, Stefan Hoffmann and Michael Wehar} 



\EventEditors{John Q. Open and Joan R. Access}
\EventNoEds{2}
\EventLongTitle{42nd Conference on Very Important Topics (CVIT 2016)}
\EventShortTitle{CVIT 2016}
\EventAcronym{CVIT}
\EventYear{2016}
\EventDate{December 24--27, 2016}
\EventLocation{Little Whinging, United Kingdom}
\EventLogo{}
\SeriesVolume{42}
\ArticleNo{23}

\begin{document}

\maketitle

%
\begin{abstract}
 The problem {\sc DFA-Intersection-Nonemptiness} asks if a given number
 of deterministic automata accept a common word. In general, this problem is \PSPACE-complete.
 Here, we investigate this problem for the subclasses
 of commutative automata and automata recognizing sparse languages.
 We show that in both cases
 {\sc DFA-Intersection-Nonemptiness} is complete for \NP\  
 and for the parameterized class $W[1]$, where the number of input
 automata is the parameter, when the alphabet is fixed.
 Additionally, we establish the same result for
 {\sc Tables Non-Empty Join}, a problem that asks if the join of several
 tables (possibly containing null values) in a database is non-empty.
 Lastly, we show that {\sc Bounded NFA-Intersection-Nonemptiness}, parameterized by the length bound,
 is $\mbox{co-}W[2]$-hard with a variable input alphabet
 and for nondeterministic automata recognizing finite strictly bounded languages,
yielding a variant leaving the realm of~$W[1]$.
\end{abstract}

\section{Introduction}

\DFAINE, or \DFAINEshort for short, is one of the most classical decision problems related to automata theory. Here, the abbreviation DFA refers to a deterministic finite automaton, a construct specified by  a tuple $\mathcal{A}=(Q,\Sigma,\delta,q_0,F)$, where $\Sigma$ is the input alphabet, $Q$ is the state alphabet, $\delta:Q\times \Sigma\to Q$ is the (total) transition function, $q_0\in Q$ is the initial state and $F\subseteq Q$ is the set of accepting states. As usual, we lift the transition function to digest words: $\delta:Q\times \Sigma^*\to Q$. Then, $L(\mathcal{A})=\{w\in\Sigma^*\mid \delta(q_0,w)\in F\}$.
More formally, this leads to the following question:
\problemdef{DFA-Intersection-Nonemptiness}{(or \DFAINEshort for short)}{A set $\mathbb{A}$ of  deterministic finite automata  with the same input alphabet $\Sigma$}{Is there a $w\in\Sigma^*$ accepted by all automata in $\mathbb{A}$?}
Following~\cite{War2001}, we define the length-bounded variation  as follows.
\problemdef{Bounded DFA-Intersection-Nonemptiness}{(or \BDFAINEshort for short)}{A set $\mathbb{A}$ of  deterministic finite automata  with the same input alphabet $\Sigma$,  $\ell\in\mathbb{N}$}{Is there a $w\in\Sigma^*$ of length  $\ell$ accepted by all automata in $\mathbb{A}$?}
Occasionally, we will also consider nondeterministic finite automata (NFA). This leads us to the problem \textsc{NFA-Intersection-Nonemptiness}, or \NFAINEshort for short, or also its bounded counterpart \BNFAINEshort.
Both \DFAINEshort and \NFAINEshort are $\PSPACE$-complete. This is also true for the bounded cases if the length~$\ell$ is given in binary. However, if the length bound is given in unary, then the classical complexity drops to $\NP$-completeness. This is also true if $|\Sigma|=1$, i.e., if we consider unary input alphabets. Then, we also call the DFA unary. It does not matter if we consider $\ell$ as an exact or an upper bound on the length of the commonly accepted string. 
These classical results can be destilled from~\cite{Koz77}.
In this paper, we will take the parameterized view on this problem. There are many natural parameters for automata problems.
\begin{itemize}
    \item $\kappa_{\mathbb{A}}$: number $|\mathbb{A}|$ of automata,
    \item $\kappa_{Q}$: upper bound on the number of states of each of the automata in~$\mathbb{A}$,
    \item $\kappa_{\Sigma}$: size of the common input alphabet~$\Sigma$,
    \item $\kappa_{\ell}$: upper bound on the length of the shortest word accepted by all automata.
\end{itemize}
One could also try to look at further parameters, like `degree of nondeterminism' or `number of final states', but we will not dive into these things in this paper.

Of course, one could also combine these parameters, leading to parameterizations like $\kappa_{\Sigma,\ell}=\kappa_{\Sigma}+\kappa_{\ell}$.
Todd Wareham~\cite{War2001} started out a study of \BDFAINEshort from a parameterized perspective. The table below shows the current picture, with references added also to other sources that later sharpened Wareham's analysis. The star ${}^*$ assumes $\PTIME\neq\NP$.
$$\begin{array}{|llllllllll|}\hline
   \kappa_{\mathbb{A}}  & \kappa_{Q} & \kappa_{\Sigma} &\kappa_{\ell} & \kappa_{\mathbb{A},Q} & \kappa_{\mathbb{A},\Sigma}& \kappa_{\mathbb{A},\ell} & \kappa_{Q,\Sigma}& \kappa_{Q,\ell}& \kappa_{\Sigma,\ell}\\
    \WNL\text{-c.} & {}\notin^*\XP              & {}\notin^*\XP    & \W[\sync]\text{-c.} & \FPT & \WNL\text{-c.} & \W[1]\text{-h.} & \FPT&\W[2]\text{-h.}& \FPT\\
    \cite{Gui2011a} &   (\cite{War2001})              & \cite{War2001}  & \cite{BruFer2020} &\cite{War2001} & \cite{Gui2011a}&\cite{War2001} &\cite{War2001} &\cite{War2001}&\cite{War2001}\\\hline
\end{array}$$
For the parameter $\kappa_{Q}$, the reference is put in parentheses, as it follows from the proof of Lemma~6 in~\cite{War2001}, where it is shown how to produce an equivalent \BDFAINEshort-instance to a given instance of \textsc{Dominating Set}, such that each DFA has two states.
This rules out $\XP$-algorithms, assuming $\PTIME\neq\NP$. In the unbounded variant, the  length parameter does not make sense. We collect the remaining results next.
\begin{propositionrep}
\begin{itemize}
    \item $(\DFAINEshort,\kappa_Q),(\DFAINEshort,\kappa_\Sigma)\notin \XP$, unless $\PTIME=\NP$;
    \item $[\DFAINEshort,\kappa_{\mathbb{A}}]^{\fpt}=[\DFAINEshort,\kappa_{\mathbb{A},\Sigma}]^{\fpt}=\XP$;
    \item $(\DFAINEshort,\kappa_{\mathbb{A},Q}),(\DFAINEshort,\kappa_{Q,\Sigma})\in\FPT$.
\end{itemize}
\end{propositionrep}
\begin{proof}
The first item is a trivial consequence of what was said about the bounded case. The last item requires a quick analysis of the proof of Theorem~8 in \cite{War2001}. The second item is possibly surprising. However, $(\DFAINEshort,\kappa_{\mathbb{A}})$ is known to be complete for the class \XNL with respect to parameterized logspace reductions. The closure of \XNL under fpt-reductions equals \XP, just because the polynomial-time reduction closure of \NL equals \PTIME; see \cite{ElbStoTan2015}.
\end{proof}

The previous result may be viewed quite natural, as similar results exist for the question if a given multihead two-way NFA accepts a given word, where the parameter is the number of heads. We know of \XNL-completeness with respect to parameterized logspace reductions for the unbounded version, while the bounded version is \WNL-complete; see \cite{ElbStoTan2015}.

In this paper, we are continuing examining this type of problem in several directions.
\begin{itemize}
    \item Some of the listed results are hardness results only. Can we find a proper home for the corresponding parameterized problem?
    \item As done with, say, graph problems that are hard on general instances, the list of (hardness) results also motivates to look into restricted classes of finite automata, or of regular languages, respectively. We will follow this idea by considering in particular commutative regular languages and sparse regular languages.
    \item Conversely, we also study the influence of determinism in finite automata regarding the intersection problem by looking at \NFAINEshort.
\end{itemize}


Combinatorial automata theory problems have been somewhat neglected from the perspective of parameterized complexity. This paper aims at proving that this can be considered as an oversight, as such problems give ample rooms for parameterized investigations.

\paragraph*{Some Further Notations and Conventions} If $n$ is some nonnegative integer, i.e., $n\in\mathbb{N}$, then $[n]=\{m\in\mathbb{N}\mid m<n\}$.
If $f:X\to Y$ is a (partial) function, then $\dom(f)=\{x\in X\mid \exists y\in Y:f(x)=y\}$ is the \emph{domain} of $f$ and $\range(f)=\{y\in Y\mid \exists x\in X: f(x)=y\}$ is the \emph{range} of~$f$.
A finite set $A$ is also called an \emph{alphabet}; then, $A^*$ is the free monoid generated by $A$, and its elements are \emph{word}s over $A$. $|w|$ is the length of the word $w$, while $|w|_a$ counts the number of occurrences of symbol $a$ in $w$.

\section{Reviewing Parameterized Complexity}

As not all parameterized complexity classes that we deal with are that well-known, we review some of them in the following. However, we assume some basic knowledge of the area, as can be found in monographs like \cite{DowFel2013,FluGro2006}.
In particular, the ``algorithmic classes'' \FPT\ and \XP\ should be familiar. Also recall that this one of the few occasions in complexity theory where an unconditional separation $\FPT\subsetneq\XP$ is known. Also, we assume knowledge about parameterized reductions, or fpt-reductions for short.

As in the classical world, a parameterized complexity class can be characterized by one (complete) problem, assuming the class is closed under parameterized reductions. Examples comprise the following classes; for the typical problems, the parameter will be always~$k$:
\begin{description}
\item[{$\W[1]$}] Given a nondeterministic single-tape  Turing machine and $k\in\mathbb{N}$, does it accept the empty word within at most $k$ steps?
\item[{$\W[2]$}] Given a nondeterministic  multi-tape  Turing machine and $k\in\mathbb{N}$, does it accept the empty word within at most $k$ steps?
\item[{$\A[2]$}] Given an alternating single-tape Turing machine whose initial state is existential and that is allowed to switch only once into the set of universal states  and $k\in\mathbb{N}$, does it accept the empty word within at most $k$ steps?
\item[{$\WNL$}] Given a nondeterministic single-tape  Turing machine an integer $\ell\geq 0$ in unary and $k\in\mathbb{N}$, does it accept the empty word within at most $\ell$ steps, visiting at most $k$ tape cells?
\item[{$\WP$}] Given a nondeterministic single-tape  Turing machine and some integer $\ell\geq 0$ in unary and $k\in\mathbb{N}$, does it accept the empty word within at most $\ell$ steps, thereby making at most $k\leq\ell$ nondeterministic steps?
\end{description}

 The \emph{Turing way} to these complexity classes (i.e., using variants of Turing machines and defining parameterized problems on these machines) is described also in  \cite{Ces2003,Gui2011a}.
However, choosing a `typical' problem $\Pi$ and defining a (parameterized) complexity class as the set of all (parameterized) problems that are   
fpt-reducible to $\Pi$ (also written as $[\Pi]^{\fpt}$ for short) clearly also works for problems not related to Turing machines. For instance, recall the problem \textsc{Clique} that asked, given a graph~$G$ and an integer~$k$, if there exists a clique (i.e., a set of pairwise adjacent vertices) of size~$k$ in~$G$. With parameterization $\kappa_{\text{sol}}(G,k)=k$, $\W[1]=[\textsc{Clique},\kappa_{\text{sol}}]^{\fpt}$. 
Another famous (and useful) problem that characterizes \W[1] is \textsc{Multicolored Clique}: Given a $k$-partite graph $G=(V,E)$, i.e., $V=\bigcup_{i=1}^kV_i$, with $V_i\cap V_j\neq\emptyset\implies i=j$, such that each $V_i$ forms an independent set, we ask for a possibility to select one vertex $v_i$ from each $V_i$ such that $\{v_1,v_2,\dots,v_k\}$ forms a clique.

For any parameterization~$\kappa$ of $\DFAINEshort$,  $[\DFAINEshort,\kappa]^{\fpt}\subseteq [\NFAINEshort,\kappa]^{\fpt}$.
In~\cite{BruFer2020}, the class $\W[\sync]$ was introduced, referring to the problem of determining, given a DFA $A$, if $A$ admits a synchronizing word of length at most~$k$, with $k$ being the parameter.
There is no need to give further details, as equivalently this class is determined by \BDFAINEshort, parameterized by~$\kappa_\ell$, i.e., $\W[\sync]=[\BDFAINEshort,\kappa_\ell]^{\fpt}$. 
Also, $\WNL=[\BDFAINEshort,\kappa_{\mathbb{A}}]^{\fpt}$. 
The mentioned (and more) parameterized complexity classes are summarized in \autoref{fig:complexity_classes_hierarchy}.

\begin{figure}[tbh]
\centering
		\begin{tikzpicture}[transform shape,scale=0.61,shorten >=1pt,node distance=1.9cm,on grid,auto, state/.style={draw, minimum size=0.75cm, fill=black!5, thick}]
		
		\node[state, white, text=black] (FPT) {\FPT};
		
		\node[state, right=of FPT, white] (Level1Empty) {};
		\node[state, below=of Level1Empty, white, text=black] (W[1]) {$\W[1]$};
		\node[state, above=of Level1Empty, white, text=black] (A[1]) {$\A[1]$};
		
		\node[state, right=of Level1Empty, white, text=black] (Level2Empty) {};
		\node[state, below=of Level2Empty, white, text=black] (W[2]) {$\W[2]$};
		
		\node[state, right=of Level2Empty, white, text=black] (W[Sync]) {$\W[\sync]$};
		
		\node[state, right=of {W[Sync]}, white, text=black] (Level4Empty) {};
		\node[state, below=of Level4Empty, white, text=black] (W[3]) {$\W[3]$};
		\node[state, above=of Level4Empty, white, text=black] (A[2]) {$\A[2]$};
		
		\node[state, right=of Level4Empty, white, text=black] (Level5Empty) {};
		\node[state, below=of Level5Empty, white, text=black] (Level5WDots) {\dots};
		\node[state, above=of Level5Empty, white, text=black] (A[3]) {$\A[3]$};
		
		\node[state, right=of Level5Empty, white, text=black] (WNL) {$\WNL$};
		\node[state, above=of WNL, white, text=black] (Level6ADots) {\dots};
		
		\node[state, right=of WNL, white, text=black] (Level7Empty) {};
		\node[state, below=of Level7Empty, white, text=black] (W[SAT]) {$\WSAT$};
		
		\node[state, right=of Level7Empty, white, text=black] (Level8Empty) {};
		\node[state, below=of Level8Empty, white, text=black] (W[P]) {$\WP$};
		\node[state, above=of Level8Empty, white, text=black] (AW[SAT]) {$\A\WSAT$};
		
		\node[state, right=of Level8Empty, white, text=black] (XPcapparaNP) {$\XP\cap \paraNP$};
		\node[state, above=of XPcapparaNP, white, text=black] (AW[P]) {$\A\WP$};
		
		\node[state, right=of XPcapparaNP, white, text=black] (XPFillnode) {};
		\node[state, right=of XPFillnode, white, text=black] (XP) {$\XP$};
		\node[state, below=of XP, white, text=black] (para-NP) {$\paraNP$};

		\path[->, >={To[inset=0pt,length=1.6pt]}]
		(FPT)
		    edge (W[1])
		    edge (A[1])
		(W[1])
		    edge[bend right=20] (A[1])
		    edge (W[2])
		(A[1])
		    edge[bend right=20] (W[1])
		    edge (A[2])
		(W[2])
		    edge (W[3])
		    edge (W[Sync])
		(W[Sync])
		    edge (A[2])
		    edge (WNL)
		    edge[bend left=6] (W[P])
		(W[3])
		    edge (Level5WDots);
		\path
		(W[3])
		    edge [draw=white,line width=3pt] (WNL);
		\path[->, >={To[inset=0pt,length=1.6pt]}]
		(W[3])
		    edge (WNL)
		    edge[draw=white,line width=3pt] (A[3])
		    edge (A[3])
		(A[2])
		    edge (A[3])
		(Level5WDots)
		    edge (W[SAT])
		(WNL)
		    edge (XPcapparaNP)
		(A[3])
		    edge (Level6ADots)
		(W[SAT])
		    edge (W[P])
		    edge[bend left=28,draw=white,line width=3pt] (AW[SAT])
		    edge[bend left=28] (AW[SAT])
		(Level6ADots)
		    edge (AW[SAT])
		(W[P])
		    edge (XPcapparaNP)
		    edge[bend left=28,draw=white,line width=3pt] (AW[P])
		    edge[bend left=28] (AW[P])
		(AW[SAT])
		    edge (AW[P])
		(XPcapparaNP)
		    edge (XP)
		    edge (para-NP)
		(AW[P])
		    edge (XP)
		;
		\end{tikzpicture}
	\caption[Complexity Classes Hierarchy]{Visualization of the complexity classes (`A $\rightarrow$ B' means `A is contained in B')}
	\label{fig:complexity_classes_hierarchy}
\end{figure}
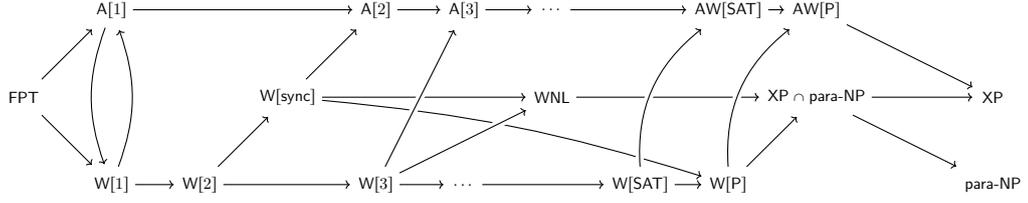

To better explain the Turing way to parameterized complexity classes, let us prove the following result that fixes the position of  $(\BDFAINEshort,\kappa_{\mathbb{A},\ell})$ within the \W-hierarchy.

\begin{proposition}$(\BDFAINEshort,\kappa_{\mathbb{A},\ell})$ is \W[1]-complete.
\end{proposition}

\begin{proof}
\W[1]-hardness was shown in~\cite{War2001}. For membership in \W[1], we explain how a single-tape NTM $M$ can solve an instance of \BDFAINEshort, specified by a finite set $\mathbb{A}$ of DFA's and a word length parameter~$\ell$. Without loss of generality, assume that all DFAs have the same state alphabet~$Q$, and that $\Sigma\cap Q=\emptyset$.  This also defines the tape alphabet $\Sigma\cup Q$ of~$M$.
The transition functions of the DFAs are hard-wired into the transition function of~$M$.
\begin{enumerate}
    \item $M$ guesses $\ell$ times (a) a letter from~$\Sigma$, (b) $|\mathbb{A}|$ many letters from~$Q$ and writes them onto its tape, which then contains a word $w$ from $(\Sigma Q^{|\mathbb{A}|})^\ell$.
    \item $M$ moves back to the beginning of the word. On this way back, $M$ checks if the last $|\mathbb{A}|$ many state letters are final states, i.e., $w[|w|]$ should be a final state of the last DFA, $w[|w|-1]$ should be a final state of the penultimate DFA, \dots, $w[|w|-|\mathbb{A}|+1]$ should be a final state of the first DFA.
    \item $M$ reads and memorizes $a_1\in\Sigma$; then, it reads (and memorizes) the first state letter~$q$ and checks if there is a transition of the initial state of the first DFA into~$q$, using the input letter~ $a_1$; then it reads the next state letter and performs a similar check for the next DFA, etc., until the last DFA is checked. If any of these checks fails, the NTM aborts, i.e., it enters an infinite loop.
    \item $M$ reads and memorizes $a_2\in\Sigma$; then, it moves $|\mathbb{A}|$ steps to the left and reads (and memorizes) the first state letter~$q$, moves $|\mathbb{A}|+1$ steps to the right and checks if the state letter $q'$ at that position corresponds to a transition $q\stackrel{a_2}\longrightarrow q'$ of the first DFA; if this check fails, the NTM aborts; otherwise, it moves $|\mathbb{A}|$ steps to the left and checks the transition of the next DFA etc.
    \item Altogether, assuming no abort happens, the NTM~$M$ reads a letter~$a_3$ from~$\Sigma$ when moving leftwards after $|\mathbb{A}|(2|\mathbb{A}|+1)$ many steps. This starts the checking of transitions using $a_3$. The whole scenario is repeated $\ell-1$ times in total.
\end{enumerate}
Putting things together, the constructed Turing machine $M$ halts in $$f(|\mathbb{A}|,\ell)=2\ell\cdot (|\mathbb{A}|+1)+(|\mathbb{A}|+1)+(\ell-1)\cdot |\mathbb{A}|\cdot (2|\mathbb{A}|+1)$$ many steps if and only if $(\mathbb{A},\ell)$ is a \yes-instance of \BDFAINEshort.
\end{proof}

The previous proof does not use any particular deterministic features of the finite automata, so that we can conclude the following result.

\begin{corollary}$(\BNFAINEshort,\kappa_{\mathbb{A},\ell})$ is \W[1]-complete.
\end{corollary}

Whether or not $(\BDFAINEshort,\kappa_{Q,\ell})$, is \W[2]-complete, remains an \textbf{open question}.
There could be further interesting complexity classes between \W[2] and
\W[\sync].

\section{Table Joins With Null Values}


In the following section, we will prove for quite a number of problems that they are complete in the parameterized complexity class \W[1]. But first, here we explain a problem originating in the theory of databases, where more details about this background are given in the appendix. This new problem shows interesting links to our further problems of study.


\begin{definition}
A \emph{binary table} (with missing information or null values) is a pair $T = (F, M)$ such that $M$ is an $m\times n$ two-dimensional array of $0$'s, $1$'s, and $*$'s and $F$ is a one-to-one function such that $\dom(F) = [n]$.  The rows of $T$ are exactly the one-dimensional rows of
~$M$.  A \emph{labeled row} of~$T$ is a pair $(F, R)$, where $R$ is a row of~$T$ and $F$ is 
from~$T$.  
\end{definition}

It is necessary to package $R$ with~$F$ so that we know how entries within a row are associated with column labels within $T$. 
Intuitively, $F$ represents the column labels and $M$ represents the stored data, where $*$'s indicate missing information, or null values.

\begin{table}[bth]
\begin{minipage}{.5\textwidth}
    \begin{center}
        \begin{tabular}{ c c c c }
         $T_1$: &A & B & C \\
         \hline
         &0 & 0 & $*$ \\
         &1 & $*$ & 1
        \end{tabular}
    \end{center}
\end{minipage}
\begin{minipage}{.5\textwidth}\begin{center}
     \begin{tabular}{ c c c }
         $T_2$: &B & C \\
         \hline
         &0 & 0 \\
         &0 & 1
        \end{tabular}\end{center}
\end{minipage}
    \caption{A tabular representation of $T_1$ and of $T_2$}
    \label{tab:binarytables}
\end{table}

\begin{example}\label{ex:twobinarytables}
Consider the following two binary tables $T_1 := (F_1, M_1)$ where $F_1 = \{ (0, \text{A}), (1, \text{B}), (2, \text{C}) \}$ and $T_2 := (F_2, M_2)$ where $F_2 = \{ (0, \text{B}), (1, \text{C}) \}$, as detailed  in \autoref{tab:binarytables}.
\end{example}

Notice that tables $T_1$ and $T_2$ have some column labels in common.  This allows us to consider whether rows from the respective tables are compatible on common columns.  This leads us to the following definition of compatibility.


\begin{definition}
Let binary tables $T_1 = (F_1, M_1)$ and $T_2 = (F_2, M_2)$ be given.  Let a labeled row $x_1 = (F_1, r_1)$ from $T_1$ and a labeled row $x_2 = (F_2, r_2)$ from $T_2$ be given.  We say that $x_1$ and $x_2$ are \emph{compatible} if for all $l \in \range(F_1) \cap \range(F_2)$, $r_1[F_{1}^{-1}(l)] = r_2[F_{2}^{-1}(l)]$ or $r_1[F_{1}^{-1}(l)] = *$ or $r_2[F_{2}^{-1}(l)] = *$.
\end{definition}

As in most circumstances, the labeling function of a labeled row is clear from the context, we can omit its mentioning then. In \autoref{ex:twobinarytables}, the first row of $T_1$ is compatible with both rows of $T_2$, but the second row of $T_1$ is compatible only with the second row of $T_2$.

Now, we are ready to our problem related to database theory.

\problemdef{Tables-Non-Empty-Join}{(or \TNEJshort for short)}{A list of $k$ binary tables $\mathcal{L} = \{T_i\}_{i \in [k]}$}{Does there exist a family of labeled rows $\{x_i\}_{i \in [k]}$ satisfying the following:
\begin{itemize}
    \item for all $i \in [k]$, $x_i$ is a labeled row of $T_i$;
    \item for all $i$, $j \in [k]$, $x_i$ is compatible with $x_j$.
\end{itemize}
\mbox{ }\\[-3ex]}

We consider $\kappa_{\text{\#Tab}}=k$ (the number of tables) to be our main parameter of interest.

The preceding problem is motivated by the following intuition.  If the $k$ rows are pairwise compatible, this would mean that the rows could be merged together to form a row in the join of the $k$ tables making the join non-empty. This is explained in more details in a particular section of the appendix.



\begin{lemma}\label{lem:non-empty join to clique}
There is a fpt-reduction from $(\TNEJshort,\kappa_{\text{\#Tab}})$ to $(\textsc{Clique},\kappa_{\text{sol}})$.
\end{lemma}

\begin{proof}
Let $k$ binary tables with at most $n$ rows each be given.  We construct a graph $G$ with at most $k \cdot n$ vertices such that the binary tables have a non-empty join if and only if $G$ contains a complete graph of size $k$.

We define $G=(V,E)$ to be a $k$-partite graph where each group of vertices is associated with a binary table so that there is one vertex for each table row. In particular, $V=\bigcup_{i\in [k]}V_i$ such that each $V_i$ is an independent set of~$G$.
Two vertices from distinct groups $V_i,V_j$ are adjacent in $G$ if and only if their corresponding rows are consistent with each other.  In other words, if vertex $x\in V_i$ is associated with row $R_x$ from table $(F_i, A_i)$ and vertex $y\in V_j$ is associated with row $R_y$ from table $(F_j, A_j)$, then $x$ and $y$ are adjacent if and only if for all \mbox{$l \in \range(F_i) \, \cap \, \range(F_j)$,} $R_x[F_i^{-1}(l)] = R_y[F_j^{-1}(l)]$ or $R_x[F_i^{-1}(l)] = *$ or $R_y[F_j^{-1}(l)] = *$.

Now, if the binary tables have a non-empty join, then there exists a choice of labeled rows (one from each table) that are consistent with each other.  This choice of rows corresponds with a choice of $k$ vertices that forms a $k$-clique.  Similarly, if there exists a $k$-clique, then there is a choice of labeled rows that are pairwise consistent with each other.  
\end{proof}


\begin{lemmarep}\label{lem:Multicolored-clique to non-empty join}
There is a polynomial-time fpt-reduction from $(\textsc{Multicolored Clique},\kappa_{\text{sol}})$ to $(\TNEJshort,\kappa_{\text{\#Tab}})$.
\end{lemmarep}

\begin{proof}
Let $G=(V_1,V_2,\dots,V_k,E)$ be an instance of \textsc{Multicolored Clique}, where the goal is to select a vertex $v_i$ from each $V_i$ such that $G[\{v_1,v_2,\dots,v_k\}]$ forms a clique. Let $G$ contain $n$ vertices in total.

We are going to construct $k$ binary tables $T_i$, one for each set $V_i$. Table $T_i$ contains $|V_i|$ rows and $n$ columns, labeled with the vertices $V=\bigcup_{i=1}^kV_i$ of $G$. This means that we can use the uniform labeling function $F=F_i$ for all tables, and we can even identify $V$ with $[n]$, so that $F$ becomes the identity. Now $T_i=(F_i,M_i)$ is defined as follows:
For row $r_{i,j}$ (corresponding to vertex $v_{i,j}\in V_i$), set $r_{i,j}[v]=0$, if there is no edge between $v_{i,j}$ and $v$; set $r_{i,j}[v]=1$, if $v_{i,j}=v$;
and set $r_{i,j}[v]=*$, if there is an edge between $v_{i,j}$ and $v$.
Now, row $r_{i,j}$ is compatible with row $r_{i',j'}$ if and only if there is an edge between $v_{i,j}$ and $v_{i',j'}$, yielding the claim.
\end{proof}

\begin{corollary} $\TNEJshort$ is \NP-complete, while
$(\TNEJshort,\kappa_{\text{\#Tab}})$ is \W[1]-complete.
\end{corollary}


These hardness results motivate to look for different parameterizations of \TNEJshort.
Observe that the overall number $n$ of table rows is clearly lower-bounded by the number~$k$ of tables. In fact, the difference $s=n-k$ (that we call the \emph{surplus}) can be seen as a good (dual) parameterization, by re-examining the  reduction of \autoref{lem:non-empty join to clique}.

\begin{corollaryrep}
The \TNEJshort problem, parameterized by the surplus, is in \FPT.
\end{corollaryrep}

\begin{proof}
An easy modification of the construction in the proof of Lemma~\ref{lem:non-empty join to clique} shows that  \TNEJshort, parameterized by the surplus, reduces to \textsc{Vertex Cover}, parameterized by solution size. Hence, any known \FPT algorithm for \textsc{Vertex Cover} can be used to solve \TNEJshort. 
\end{proof}


\section{Complexity Results for Commutative Automata}



A DFA is said to be \emph{commutative} if its transition monoid is commutative
. Equivalently, a DFA $\mathcal{A}=(\Sigma,Q,\delta,q_0,F)$ is commutative if $\delta(q, ab) = \delta(q, ba)$ for all $q \in Q$, $a,b \in \Sigma$. 
A regular language is called commutative if it is accepted by a commutative DFA. 
A very special case of commutative DFA's are unary DFA's (accepting languages over a singleton input alphabet) 
that we will consider first. In order to refer to the problems in the mentioned restricted settings, we prefix the problem names with $\unary$ and with $\comm$.

\subsection{Classical Complexity Results}





\begin{theoremrep}
\label{thm:unary-dfa-np-complete}
The 
\unary\DFAINE problem 
is $\NP$-complete.
\end{theoremrep}
The hardness part of this result follows by a reduction from 3\textsc{SAT} based on related reductions in \cite{Stockmeyer73} and \cite{Galil76}.
For more details, we refer to the appendix. We now generalize this result towards general commutative DFA's.

\begin{toappendix}

As the proof of Theorem~\ref{thm:unary-dfa-np-complete} (concerning \NP-hardness) is a bit challenging to extract from the classical sources, we provide a self-contained proof of it in the following.

\begin{proof}[Proof of the \NP-hardness part of Theorem~\ref{thm:unary-dfa-np-complete}] We show a reduction from 3\textsc{SAT}.
Let $F$ be a given CNF formula with variables $x_1,\dots,x_n$. $F$ consists of the clauses $c_1,\dots,c_m$.
After a little bit of cleanup, we can assume that each variable occurs at most once in each clause.
Let $p_1,\dots,p_n$ be the first $n$ prime numbers. It is known that $p_n\leq (n\ln (n\ln n)))$.
To simplify the following argument, we will only use that $p_n\sim n\ln n$, as shown in \cite[Satz, page 214]{Lan09}. 
If a positive integer $z$ satisfies
$$\forall i:1\leq i\leq n\implies z\equiv 0\bmod p_i \lor z\equiv 1\bmod p_i,$$
then $z$ represents an assignment $\alpha$ to $x_1,\dots, x_n$ with $\alpha(x_i)=z\bmod p_i$.
Then, we say that $z$ satisfies $F$ if this assignment satisfies $F$.
Clearly, if $z\in \{0\}^j\{0^{p_k}\}^*$ for some $2\leq j<p_k$, then $z$ cannot represent any assignment, as $z\bmod p_k$ is neither $0$ nor $1$.
There is a DFA $A_{x_i}$ for $L_{p_i}:=\{0,\varepsilon\}\{0^{p}\}^*$ with $p_i+1\leq p_n$ states, where $p_i\in\{p_1,\dots,p_n\}$. Notice that $x\in 0^*$ can be interpreted as an assignment if and only if $x\in\bigcap_{i=1}^nL_{p_i}$.
\par
To each clause $c_j$ with variables $x_{i_j(1)},\dots, x_{i_j(|c_j|)}$ occurring in $c_j$, for a suitable injective index function $i_j:\{1,\dots,|c_j|\}\to \{1,\dots,n\}$, 
there is a unique assignment $\alpha_j\in\{0,1\}^{|c_j|}$ to these variables that falsifies $c_j$.
This assignment can be represented by the language 
$$L_{\alpha_j}:=\{0^{z_{k_j}}\}\cdot \{0^{p_{i_j(1)}\cdots p_{i_j(|c_j|)}}\}^*
$$
with $0\leq z_{k_j}< p_{i_j(1)}\cdots p_{i_j(|c_j|)}$ being uniquely determined by  $z_{k_j}\equiv \alpha(r)\bmod p_{i_j(r)}$ for $r=1,\dots,|c(j)|$.
As $p_{i_j(1)}\cdots p_{i_j(|c_j|)}\leq p_n^3$ (3\textsc{SAT}), $0^*\setminus L_{\alpha_j}$ can be accepted by a DFA $A_{c_j}$ with at most $p_n^3$ states.
Altogether, we produced $n+m$ many DFA's such that
$\bigcap_{i=1}^n L(A_{x_i})\cap \bigcap_{j=1}^mL(A_{c_j})\neq\emptyset$ if and only if $F$ has a satisfying assignment.
\end{proof}
\end{toappendix}

\newcommand{\sort}{\mathrm{sort}}

Also in order to build a  link to \autoref{sec:sparse}, we now consider an \emph{ordered input alphabet}, i.e., an alphabet with a linear order on it. Then, $\sort(x)$ refers to the string obtained by sorting the characters of $x$, maintaining multiplicities.

\begin{proposition}\label{prop:comm-ordered}
Let a commutative DFA $\mathcal{A}$ over an ordered input alphabet~$\Sigma$ be given.  For all strings $x\in\Sigma^*$, $\mathcal{A}$ accepts $x$ if and only if $\mathcal{A}$ accepts $\sort(x)$.
\end{proposition}


\begin{theorem}
The \comm\DFAINEshort problem 
is $\NP$-complete.
\end{theorem}
This claim is also true for \comm\NFAINEshort, see~\cite{AFHHJOW2021}, but we give here a simpler proof.
\begin{proof}
$\NP$-hardness is a direct consequence of Theorem~\ref{thm:unary-dfa-np-complete}, because every unary DFA is a commutative DFA.
In order to see $\NP$-membership, we prove the following \underline{claim}: If $\mathbb{A}$ is a collection of $k$  commutative DFA's over the common input alphabet $\Sigma$, each with no more than $s$ states, such that $L:=\bigcap_{\mathcal{A}\in\mathbb{A}}L(\mathcal{A})\neq\emptyset$, then $L\cap \Sigma^{\leq r}\neq\emptyset$ for $r:=|\Sigma|\cdot s^k$.

Namely, let the indices of $\Sigma=\{a_1,a_2,\dots,a_n\}$ reflect an arbitrary linear order on $\Sigma$. Assume that $L\neq\emptyset$, so that some $w\in L$ exists. Consider some $w\in L$ of shortest length. 
Then, $\sort(w)\in L$ (see \autoref{prop:comm-ordered}) has also minimal length. Assume $r<|w|=|\sort(w)|$.  Now, $\sort(w)=a_1^{m_1}a_2^{m_2}\cdots a_n^{m_n}$. As $r<|\sort(w)|$, we have $m_i>s^k$ for some $i-1\in [n]$. Then, in the product automaton constructed from all $\mathcal{A}\in\mathbb{A}$, a state must repeat upon digesting the factor $a_i^{m_i}$ of $\sort(w)$, so that we arrive at a shorter word in $L$ by cutting out this loop, contradicting the assumed minimality of $|w|$. This proves the claim.

Hence, it is sufficient to guess a polynomial number of bits to obtain a candidate string $v\in\Sigma^*$ of which membership in $L$ can be checked by sequentially simulating all DFA's.
%
%
\end{proof}

The previous result means that there is no need to consider the  variation \BDFAINEshort, because there are always certificates of polynomial size for the length parameter.

\subsection{Parameterizations} 

While $(\unary\DFAINEshort,\kappa_\ell)\in \FPT$, this is no longer the case for commutative DFAs. However, we do not know the exact classification of this problem; we leave this  \textbf{open}.

\begin{propositionrep}\label{prop:commutative-W2}The following inclusions are immediate:\\
$\W[2]\subseteq [(\comm\DFAINEshort,\kappa_{Q,\ell})]^{\text{fpt}}\subseteq [(\comm\DFAINEshort,\kappa_\ell)]^{\text{fpt}}\subseteq \W[\sync]$.
\end{propositionrep}
The proof of this proposition can be derived by re-visiting the \W[2]-hardness proof of Wareham for $(\DFAINEshort,\kappa_{Q,\ell})$.
A simplified version of it can be found in the appendix.
The same proof also shows that $(\comm\DFAINEshort,\kappa_Q)$ is not in $\XP$, unless $\PTIME=\NP$.
\begin{toappendix}
\begin{proof}
We have to prove that $(\comm\DFAINEshort,\kappa_{Q,\ell})$ is \W[2]-hard. Consider the well-known \W[2]-hard problem \textsc{Hitting Set}, parameterized by solution size. An instance of \textsc{Hitting Set} consists in a hypergraph $G=(V,E)$ and a positive integer~$k$. The task is to find a selection $v_1,\dots,v_k$ from $V$ such that each hyperedge $e\in E$ contains some $v_i$.
Now, construct $|E|$ many 2-state DFAs with input alphabet $V$. DFA $\mathcal{A}_e$ (for $e\in E$) loops in its initial state with all letters but those from $e$, with which it moves to its final state, which is a sink state.
\end{proof}
\end{toappendix}

We now parameterize the \DFAINE problem by the number of automata 
provided in the input. We will prove that $(\unary(\textsc{B})\DFAINEshort,\kappa_{\mathbb{A}})$ is \W[1]-complete.
While in the general case (of arbitrary alphabets), it is unclear if $(\BDFAINEshort,\kappa_\ell)$ is easier than $(\BDFAINEshort,\kappa_{\mathbb{A}})$, as it is unknown (and with unknown consequences) if \W[\sync] is a proper subset of \WNL, the situation is clearer in the unary case, as $(\unary\BDFAINEshort,\kappa_\ell)\in \FPT$ is then trivial. Under the Exponential Time Hypothesis (implying $\FPT\neq\W[1]$), a separation from $(\unary\BDFAINEshort,\kappa_{\mathbb{A}})$  is produced.

\begin{lemmarep}\label{lem:CliqueToDFA}
There is a polynomial-time computable fpt-reduction from $(\textsc{Clique},\kappa_{\text{sol}})$ to $(\unary\DFAINEshort,\kappa_{\mathbb{A}})$, so that 
$\W[1]\subseteq[\unary\DFAINEshort,\kappa_{\mathbb{A}}]^{\text{fpt}}$.
\end{lemmarep}

The claimed \W[1]-hardness of $(\unary\DFAINEshort,\kappa_{\mathbb{A}})$ also follows by \cite[Corollary~11]{MorRehWel2020}, but for the sake of self-containment, we sketch a proof in the appendix. As in~\cite{MorRehWel2020}, we make use of number-theoretic tools, in particular of the Generalized Chinese Remainder Theorem.

\begin{proof}
Let a graph $G=(V,E)$ with $n$ vertices be given. Let $V=[n]$. We build $k$ choose 2 DFA's such that the DFA's have a non-empty intersection if and only if the graph~$G$ has a $k$-clique. We first need to find $k$ natural numbers $m_1,m_2,\dots,m_k$ such that:
\begin{itemize}
\item the numbers are at least $n$ each,
\item the numbers are $\Oh(n\log(n))$ (i.e., not too big),
\item the numbers are all pairwise relatively prime.
\end{itemize}
We can find such numbers by building a prime sieve.  This is guaranteed to work for $n$ sufficiently large by the prime number theorem.

Each of the $k$ choose 2 DFA's will have a cycle length that is a product of two of the numbers. A remainder modulo a number will represent a vertex in the graph (namely, if $z\bmod m_i\in[n]$) or an invalid remainder. Now, a state~$z$ in a DFA, specified, say, by $m_i$ and~$m_j$, will represent two remainders, one with respect to $m_i$ and the other one with respect to $m_j$.  It is a final state if both remainders are valid and the associated vertices  $z\bmod m_i$ and  $z\bmod m_j$  have an edge in~$G$.
\end{proof}

Now, we add $\Sigma\text{-}$ to our problem abbreviations to fix the input alphabet to $\Sigma$. By~\autoref{prop:comm-ordered}, we know that, given an ordered alphabet $\Sigma$ and a collection $\mathbb{A}$ of DFA's, then $x\in\bigcap_{\mathcal{A}\in\mathbb{A}}L(\mathcal{A})$ if  and only if $\sort(x)\in\bigcap_{\mathcal{A}\in\mathbb{A}}L(\mathcal{A})$.
Therefore, the next statement is actually a corollary of our reasoning on (strictly) bounded languages in the next section.

\begin{theorem}\label{thm:commSigmaDFAINE-W1}
$(\comm\Sigma\text{-}\DFAINEshort,\kappa_{\mathbb{A}})$ is $W[1]$-complete for  any \textit{fixed alphabet}~$\Sigma$.
\end{theorem}

\section{Automata Recognizing Sparse  Languages}
\label{sec:sparse}

A language $L \subseteq \Sigma^*$ is \emph{sparse} if the function that counts the number of strings of length $n$ in the language is bounded by a polynomial function of $n$.
Sparse languages were introduced into computational complexity
theory by Berman \& Hartmanis~\cite{DBLP:journals/siamcomp/BermanH77}.

 A language $L \subseteq \Sigma^*$ is \emph{bounded}
if there exist words $w_1, w_2,\ldots, w_m \in \Sigma^*$
such that $L \subseteq w_1^*w_2^* \cdots w_m^*$. This language
class was introduced in~\cite{GinSpa64}.
Both sparse and bounded languages generalize unary languages.
By definition, every bounded language is sparse. Conversely,
every regular sparse language is bounded; this is even true for context-free languages, see~\cite{DBLP:journals/eik/LatteuxT84}.
A language is called \emph{strictly bounded} if it is a subset of $a_1^*a_2^* \cdots a_m^*$
for an alphabet $\Sigma = \{a_1,a_2, \ldots, a_m\}$ of size $m$.
if there exists a sequence of not necessarily distinct
$a_1, a_2,\ldots, a_m \in \Sigma$ such that it is contained in $a_1^* a_2^*\cdots a_m^*$.
For example, if $\Sigma = \{a,b\}$, then $a^* b^*$ is strictly bounded,
but $a^* b^* a^*$ is only letter-bounded, not strictly bounded.

As before, we use prefixing, now with $\sparse$ and with $\sbounded$, in order to refer to the corresponding restrictions of automata models in our problem definitions.

\subsection{Combinatorics and Classical Complexity}

\begin{toappendix}
\begin{theorem}
 The intersection and union of bounded languages is a bounded language.
\end{theorem}
\begin{proof}
 For intersection, we show the stronger claim that intersecting
 a bounded language with any language gives a bounded language.
 Let $U \subseteq w_1^*w_2^* \cdots w_n^*$
 and $V \subseteq \Sigma^*$.
 Then, $U \cap V \subseteq U \subseteq w_1^*w_2^* \cdots w_n^*$.
 The union was stated in~\cite{GinSpa64}.
\end{proof}

\begin{theorem}[\cite{GinsburgSpanier66}]
\label{thm:bounded_regular_form}
 A language $L \subseteq w_1^*w_2^* \cdots w_k^*$ is regular if and only if
 it is a finite union of languages of the form $L_1L_2 \cdots L_k$, where each $L_i \subseteq w_i^*$ is regular.
\end{theorem}
\end{toappendix}

\begin{toappendix} 
 \begin{lemma}[\cite{DBLP:conf/cai/Hoffmann19}]
\label{lem::unary_single_final}
  Let $L \subseteq \{a\}^{\ast}$ be a unary language that is recognized
  by an automaton with a single final state, index $i$ and period $p$.
  Then either $L = \{u\}$ with $|u| < i$ (and if the automaton is minimal we would have $p = 1$),
  or $L$ is infinite with $L = a^{i+m}(a^p)^{\ast}$ and $0 \le m < p$. Hence
  two words $u,v$ with $\min\{|u|, |v|\} \ge i$ are both in $L$ or not if and only
  if $|u| \equiv |v| \pmod{p}$.
 \end{lemma}
\end{toappendix}

%

An incomplete (or partial) automaton (PDFA for short) is called \emph{cyclic} if
its state diagram
is a single directed cycle.
If $\mathcal A = (\Sigma, Q, \delta, q_0, F)$
is a cyclic automaton, then, for any $q \in Q$,
we have $\{ w \in \Sigma^* \mid \delta(q, w) = q \} \subseteq u_q^*$
for some word $u_q \in \Sigma^*$.
By considering multiple final states and initial tails,
we find that languages accepted by cyclic DFA's
are finite unions of the form $uv^* \subseteq u^* v^*$,
hence they are bounded languages by the closure properties above.

Based on this, we call an automaton \emph{polycyclic}
if every strongly connected component consists of a single cycle, where
we count parallel transitions multiple times and do not collapse
them to a single transition. By adding an exceptional trash state, we could easily turn (poly-)cyclic automata into complete automata.

 \begin{example} A few examples and non-examples of polycyclic automata.
 
 \begin{center}
  \begin{tikzpicture}[node distance=15mm, auto]
  \tikzstyle{vertex}=[circle,fill=black!25,minimum size=10pt,inner sep=0pt]
   
   \node at (0.5,0) {};
  
   \foreach \name/\x in {s/1.5, 2/2.25}
    \node[vertex] (G-\name) at (\x,0) {};
    
   \foreach \name/\angle/\text in {P-1/180/5, P-2/108/6, 
                                   P-3/36/7, P-4/-36/8, P-5/-108/9}
    \node[vertex,xshift=4cm] (\name) at (\angle:0.9cm) {};

  \foreach \name/\angle/\text in {Q-1/90/5, Q-2/180/6, 
                                  Q-3/0/7}
    \node[vertex,xshift=2cm,yshift=-1.9cm] (\name) at (\angle:0.6cm) {};
    
  \node[vertex] (S1) at (5.7,0.5) {};
  \node[vertex] (S2) at (6.5,0.5) {};
  
  \foreach \from/\to/\label in {G-s/G-2/a,G-2/P-1/b,P-1/P-2/a,P-2/P-3/c,P-3/P-4/c,P-4/P-5/c,P-5/P-1/a,P-5/Q-1/b,Q-1/Q-3/a,Q-3/Q-2/b,Q-2/Q-1/b}
     \path[->] (\from) edge node {$\label$} (\to);
     
  \path[->] (P-3) edge node {$a$} (S1);
  \path[->] (S1) edge node {$b$} (S2);
  \path[->] (S2) edge [loop below] node {$a$} (S2);
   \node at (5,-1.7) {polycyclic PDFA};
  
   \node[vertex] (S1) at (9,0.5) {};
   \path[->] (S1) edge [loop above] node {\emph{$a,b$}} (S1);
   \node at (9,0) {not polycyclic};

    \foreach \name/\angle/\text in {T-1/90/5, T-2/180/6, 
                                  T-3/0/7}
    \node[vertex,xshift=9cm,yshift=-1.6cm] (\name) at (\angle:0.85cm) {};
    
    \path[->] (T-1) edge node {\emph{$a,b$}} (T-3)
              (T-3) edge node [above]  {$a$} (T-2)
              (T-2) edge node  {$a$} (T-1);
    
      \node at (9,-2) {not polycyclic};
 \end{tikzpicture} 
 \end{center}
 \end{example}    
 
 
 \begin{theorem} \cite{Hoffmann2021b}
  A regular language is sparse iff 
  it is accepted
  by a polycyclic automaton.
 \end{theorem}

 For our \NP-containment result, we need a polynomial bound on the number
 of words $w_1, w_2,\ldots, w_m$ such that $L \subseteq w_1^* w_2^* \cdots w_m^*$
 for a sparse language $L \subseteq \Sigma^*$.
 
\begin{propositionrep}
\label{prop:number_of_words_polynomial}
 Let $L \subseteq \Sigma^*$ be recognized by a polycyclic automaton
 with $n$ states. Then, $L \subseteq w_1^*w_2^* \cdots w_m^*$
 for a number $m$ polynomial in $n$.
\end{propositionrep}
\begin{proof}
 We show the claim for an automaton with a single 
 final state. As for two languages $U \subseteq u_1^*u_2^* \cdots u_r^*$
 and $V \subseteq v_1^*v_2* \cdots v_s^*$
 we have $U \cup V \subseteq u_1^*u_2^* \cdots u_r^* v_1^*v_2* \cdots v_s^*$,
 and every language recognized with more than one final state
 is a union of at most $n$ languages recognized with a single
 final state, this implies the general claim.
 Now, first note that we have at most polynomially in $n$
 many paths from the start state to the final state
 without a repetition.
  This follows, as if we consider the adjacency matrix
 of the underlying automaton graph the number
 of these paths is upper-bounded by the sum
 of the entries in the row and column corresponding
 to the start and final state of the first, second up to the $n$-th
 power of this matrix, as the $i$-th power of
 this matrix counts the number of (not necessarily simple) paths of length $i$
 from the start state to the final state.
 This sum is polynomial in $n$, as the entries of the matrices
 involved are computed with polynomially many steps, 
 only involving multiplication and addition.
 Now, by the definition property
 of polycyclic automata, every repetition of states
 is induced by the repetition of a single word dependent on the state.
 
 Now, consider such a path $q_0, q_1, \ldots, q_r$
 with label $w$. 
 Let $L_{q_i}$ be the language in the polycyclic
 automaton we get if we start with the state $q_i$
 and final state set $\{q_i\}$. Also, write $w = w_1\cdots w_r$ with $w_i \in \Sigma$.
 Then, the language
 \[
  L_{q_0} w_1 L_{q_1} w_2 \cdots L_{q_{r-1}} w_r L_{q_r}
 \]
 is in $L$ and, in fact, $L$
 is the union of all these languages over all paths
 considered above.
 Now, under our assumptions,
 these languages are contained
 in $w_{q_0}^* w_1^* \cdots w_r^* w_{q_r}^*$, 
 and as we have polynomially many of them,
 concatenating them all gives the claim.
\end{proof}

We need the following technical lemma.

\begin{lemmarep}
\label{lem:minimize_sparse}
 For a given sparse regular language accepted by a DFA with $n$ states,
 we can compute in polynomial time a polycyclic automaton
 accepting it with at most $n$ states.
\end{lemmarep}
\begin{proof}
  A sparse language necessarily has to have a non-final trap state 
  in every complete automaton accepting it (if not, we could append to every word
  a word of length at most $n$ that gives an accepted word, which would yield
  an exponential growth for the number of words of some given length).
  Also, the sparseness
  assumption would be violated if we have a strongly connected component that lies on an accepting
  path, i.e., its states are reachable and co-reachable, but does not form a single cycle (as then we can generate exponentially many words for some length).
  Combining both observations, we see that we can leave out the non-final trap
  and this gives us a polycyclic automaton.
\end{proof}

Next we show that $\DFAINEshort$ is in \NP for automata accepting sparse language.
Note that we mean automata as complete deterministic automata here, i.e.,
we do not mean the input automata to be polycyclic.

\begin{theorem}\label{thm:sparseDFAINE-NP}
 The problem $\sparse\DFAINEshort$ 
 is $\NP$-complete.
\end{theorem}
\begin{proof}
  Let $\mathcal A_i = (Q_i, \Sigma, \delta_i, q_i, F_i)$
  for $i \in \{1,2,\ldots,k\}$ be $k$ DFA's recognizing
 sparse languages with each automaton having at most $n$ states. 
 By Lemma~\ref{lem:minimize_sparse}
 we can assume these automata are polycyclic automata.
  
  Now, consider $\mathcal A_1$.
  Then, by Proposition~\ref{prop:number_of_words_polynomial},
  there exist $w_1,w_2, \ldots, w_m \in \Sigma^*$ with $m$ polynomial in $n$.
  We have
  $
   \bigcap_{i=1}^k L(\mathcal A_i) \subseteq L(\mathcal A_1) \subseteq w_1^*w_2^* \cdots w_m^*.
  $
  So, a common word has the form
  $w = w_1^{r_1} w_2^{r_2}\cdots w_m^{r_m}$.
  Now, let $i \in \{1,2,\ldots,m\}$
  and consider the prefix $w_1^{r_1} \cdots w_{i-1}^{r_{i-1}}$
  and the $k$ states in which the $k$ automata 
  are after reading this prefix.
  When reading $w_1^{r_1} \cdots w_{i-1}^{r_{i-1}}w_i^j$
  for $j = 0,1,\ldots$, these states have to
  loop after $n^k$ many steps.
  Hence, we can suppose that
  in a common word, we have $r_i \le n^k$.
  So, we can code these numbers in binary, which only requires polynomial space
  in~$k$ and in the sizes of the automata. Then, we can guess them non-deterministically
  and 
  the verification could be carried out in polynomial time.
  That the problem is \NP-hard is implied, as unary DFA's
  are polycylic.
\end{proof}

Next, as a special case, we consider strictly bounded languages.
Notice that by assuming an ordering on the alphabet, there is a natural bijection between strictly bounded regular languages and commutative regular languages, also compare \autoref{prop:comm-ordered}. However, this does not necessarily translate the complexity results that we study in this paper.


 As the problem was $\NP$-hard for unary DFA's, we get the next corollary.

\begin{corollary}
The problem \sbounded\DFAINEshort
 is $\NP$-complete.
\end{corollary}

\subsection{Parameterized Complexity}

Apart from different interesting parameterizations, we again encounter the case that some parts of the input are considered fixed. 


\begin{propositionrep}\label{prop:sbounded-W2} The following inclusion chain holds:\\
$\W[2]\subseteq [(\sbounded\BDFAINEshort,\kappa_{\ell})]^{\fpt}\subseteq [(\bounded\BDFAINEshort,\kappa_\ell)]^{\fpt}\subseteq \W[\sync]$.
\end{propositionrep}

The proof of this proposition can be seen as a modification of  the \W[2]-hardness proof of Wareham for $(\DFAINEshort,\kappa_{Q,\ell})$. 
Observe that the automata that result from this proof are also partially ordered in the sense of \cite{DBLP:journals/jcss/BrzozowskiF80}.

\begin{toappendix}
\begin{proof}
We have to prove that $(\sbounded\DFAINEshort,\kappa_\ell)$ is \W[2]-hard. Consider the well-known \W[2]-hard problem \textsc{Hitting Set}, parameterized by solution size. An instance of \textsc{Hitting Set} consists in a hypergraph $G=(V,E)$ and a positive integer~$k$. Define an arbitrary ordering on $V$, reflected by the indices $V=\{v_1,v_2,\dots,v_n\}$.
The task is to find a selection $v_{i_1},v_{i_2},\dots,v_{i_k}$ from $V$ such that each hyperedge $e\in E$ contains some $v_{i_j}$.
Now, construct $|E|$ many  DFAs with input alphabet $V$ and state set $(V\cup\{v_0\})\times\{0,1\}$, where $V\times\{1\}$ is the set of final states and $(v_0,0)$ is the initial state. Upon reading symbol $v_j$, DFA $\mathcal{A}_e$ (for $e\in E$) moves from state $(v_i,0)$ (with $i<j$) into $(v_j,0)$ if $v_j\notin E$, but into $(v_j,1)$ if $v_j\in E$; also, $\mathcal{A}_e$ moves from state $(v_i,1)$ (with $1\leq i<j$) into $(v_j,1)$. All other transitions go into the sink state $(v_0,1)$. Obviously, $L(\mathcal{A}_e)=\{w\in v_1^*v_2^*\cdots v_n^*\mid \{v\in V\mid |w|_v=1\}=e\}$.
\end{proof}
 \end{toappendix}


However, if we have a fixed set of words $w_1,w_2, \ldots, w_m$ for which we draw languages
from $w_1^*w_2^*\cdots w_m^*$, then the problem is $\W[1]$-complete, as we prove next. The details of its quite delicate proof can be found in the appendix. This result has a number of consequences, as we clarify first.

\begin{theorem}\label{thm:sparse-DFAINE-W1}
 Let $w_1, w_2,\ldots, w_m$ be fixed words.
 Then, $\DFAINEshort$ for automata accepting languages in $w_1^*w_2^* \cdots w_m^*$, parameterized by $\kappa_{\mathbb{A}}$,
 is $\W[1]$-complete.
\end{theorem}
As strictly bounded languages are a special case, we can conclude:

\begin{corollary}\label{cor:sbounded-DFAINE-W1}
For each fixed alphabet~$\Sigma$, 
$(\sbounded\Sigma\text{-}\DFAINEshort,\kappa_{\mathbb{A}})$ is $\W[1]$-complete.
\end{corollary}
Hence, we can also deduce the last result from the previous section, which is \autoref{thm:commSigmaDFAINE-W1}.

The following sketch concerns \autoref{cor:sbounded-DFAINE-W1}, as this avoids several intricacies found in the proof of the main theorem of this section, which is \autoref{thm:sparse-DFAINE-W1}. We refer to the appendix.
\begin{proofsketch}
\W[1]-hardness follows with \autoref{lem:CliqueToDFA}. For membership, we present a reduction to \textsc{Multicolored Clique}. To this end, let $\Sigma=\{a_1,a_2,\dots,a_m\}$ be an ordered alphabet. Let $\mathcal{A}_1$, $\mathcal{A}_2$, \dots, $\mathcal{A}_k$ be minimal  DFA's accepting subsets $L_i$ ($1\leq i\leq k$) of $a_1^+a_2^+\cdots a_m^+$. Ignoring a possible trash state, the state set $Q_i$ of $\mathcal{A}_i$ (with at most $n$ states) can be decomposed into $Q_i^1$, $Q_i^2$, \dots, $Q_i^m$ such that \begin{itemize}
    \item all transitions between states within $Q_i^j$ are labelled $a_j$;
    \item transitions leaving states from $Q_i^j$ must enter $Q_i^{j+1}$ and must be labelled $a_{j+1}$;
    \item transitions entering states from $Q_i^j$ must leave $Q_i^{j-1}$ and must be labelled $a_{j}$.
\end{itemize}
This is mainly due to our assumption that $L_i\subseteq a_1^+a_2^+\cdots a_m^+$. Observe that in general $L_i\subseteq a_1^*a_2^*\cdots a_m^*$ is possible, but this just leads to more cases. Let $\Sigma^{ j,\leq n}=\{w\in  a_1^+a_2^+\cdots a_j^+\mid |w|_{a_r}\leq n \text { for } r\leq j\}$. Let  $\Sigma^{\leq m,\leq n}=\bigcup_{j=1}^m\Sigma^{\leq j,\leq n}$.
Notice that $|\Sigma^{\leq m,\leq n}|\leq mn^m$ is a polynomial, as $m$ is fixed.

The graph $G=(V,E)$ that we construct is $(mk+1)$-partite. $V=\bigcup_{i=1}^k\{i\}\times \Sigma^{\leq m,\leq n}\cup\{f\}$.
Consider $\phi:\{1,2,\dots,k\}\times \Sigma^{\leq m,\leq n}\to Q_i, (i,w)\mapsto \delta_i(q_i,w)$, where $q_i$ is the initial state of $\mathcal{A}_i$ and $\delta_i$ is the transition function of $\mathcal{A}_i$. 
Then, the $mk$ independent sets of $V\setminus\{f\}$ are just the sets $V_{i,j}=\phi^{-1}(Q_i^j)$. The edge set decomposes into
\begin{itemize}
    \item $E_1$: connect $u\in V_{i,j}$ with $v\in V_{i',j'}$ if $i\neq i'$ and $j\neq j'$;
    \item $E_2$: connect $u\in V_{i,j}$ with $v\in V_{i,j'}$ if $|j-j'|\geq 2$;
    \item $E_3$: connect $(i,w)\in V_{i,m}$ with $f$ if $\phi(i,w)\in F_i$, the final state set of $\mathcal{A}_i$; always connect $(i,w)\in V_{i,j}$ with $f$ if $j<m$.
    \item $E_4$: connect $(i,w)\in V_{i,j-1}$ with $(i,w'a_j^r)\in V_{i,j}$ if $\delta_i(\phi(i,w),a_j^r)=\phi(i,w'a_j^r)$ for $w,w'\in \Sigma^{ j-1,\leq n}$;
    \item $E_5$: connect $(i,w)\in V_{i,j}$ with $(i',w)\in V_{i',j}$ (for $i\neq i'$) if either $(i,w)$ and $(i',w)$ are both cycle vertices and they agree on their cycle equation, or $(i,w)$ is a tail vertex whose equation also satisfies that of  $(i',w)$.
\end{itemize}
Here,  $(i,w)\in V_{i,j}$ is a cycle vertex if $w=va_j^{r_{i,j}}$ for some $v\in  V_{i,j-1}$ (with $V_{i,0}=\{\varepsilon\}$) and there is a smallest possible $N(i,w)>0$ such that $\phi(i,w)=\phi(i,wa_j^{N(i,w)})$. The cycle equation would then be $x_j\equiv r_{i,j}\pmod{N(i,w)}$ and $x_j\geq r_{i,j}$. If $(i,w)\in V_{i,j}$ is not a cycle vertex, it is a tail vertex and this results in the tail equation $x_j=r_{i,j}$.  Cycle vertices $(i,w)$ and $(i',w)$ agree on their cycle equation if $r_{i,j}\equiv r_{i',j}\pmod{\gcd(N(i,w),N(i',w))}$. The claim that $G=(V,E)$ has a $(mk+1)$-clique if and only if the intersection $\bigcap_{i=1}^kL_i$ is not empty follows by applying the Generalized Chinese Remainder Theorem.

One can finally get away with the assumption that $L_i\subseteq a_1^+a_2^+\cdots a_m^+$ by considering all $2^m$ cases if (or if not) $|w|_{a_j}=0$ is possible. As $m$ is constant, it is possible to combine all these cases into a single big graph that has a multicolored clique if and only if $\bigcap_{i=1}^kL_i\neq \emptyset$.
\end{proofsketch}

\section{Intersecting Nondeterministic Finite Automata}

In this section, we consider various parameterizations of \BNFAINEshort.
This leads us to the following table that summarizes the state of the art and also points to several open classification problems.
$$\scalebox{.81}{$\begin{array}{|llllllllll|}\hline
   \kappa_{\mathbb{A}}  & \kappa_{Q} & \kappa_{\Sigma} &\kappa_{\ell} & \kappa_{\mathbb{A},Q} & \kappa_{\mathbb{A},\Sigma}& \kappa_{\mathbb{A},\ell} & \kappa_{Q,\Sigma}& \kappa_{Q,\ell}& \kappa_{\Sigma,\ell}\\
    \WNL\text{-h.} & {}\notin^*\XP              & {}\notin^*\XP    & \W[\sync]\text{-h.} & \FPT & \WNL\text{-h.} & \W[1]\text{-c.} & \FPT&\W[2]\text{-h.}& \FPT\\
    \co\W[2]\text{-h.}&&&
    \in\WNL\cap\A[3]&&&&&&\\
    \cite{Gui2011a}, \text{Propos.~\ref{prop:nfa_emptiness_w2}} &   (\cite{War2001})              & \cite{War2001}  & \cite{BruFer2020}, \text{Lemmas~\ref{lem:NFAintersec_ell_WNL}\&\ref{lem:NFAintersec_ell_A3}} &\cite{War2001} & \cite{Gui2011a}&\cite{War2001} &\cite{War2001} &\cite{War2001}&\cite{War2001}\\\hline
\end{array}$}$$

\begin{lemma}\label{lem:NFAintersec_ell_WNL}
$(\BNFAINEshort,\kappa_{\ell})\in\WNL$.
\end{lemma}

\begin{proof}
 (Sketch) Let $\mathbb{A}$ be a collection of NFAs and $\ell\in\mathbb{N}$.
We can design a nondeterministic one-tape Turing machine $M$ (starting on the empty input) that first guesses a string $w$ of length $\ell$ and then verifies, for each automaton $\mathcal{A}\in\mathbb{A}$, that $\mathcal{A}$ accepts~$w$. Only if all these simulations succeed, $M$ will accept. As required, $M$ needs only space~$\ell$. 
\end{proof}

However, it is not clear to us how we could limit the number of guessing steps in the simulation of the previous proof, i.e., it is an \textbf{open question} if $(\BNFAINEshort,\kappa_{\ell})$ belongs to $\WP$. Also, it is an  \textbf{open question} if $(\BNFAINEshort,\kappa_{\ell})$ belongs to $\A[2]$, because it is not clear how to make this simulation work with only one switch from existential to universal states. Without going into details here, let us mention that with one further switch from universal to existential states, an alternating Turing machine could be designed that simulates a given \BNFAINEshort instance in time $f(k)$, which gives (see \cite{FluGro2006} for details) the following result.

\begin{lemma}\label{lem:NFAintersec_ell_A3}
$(\BNFAINEshort,\kappa_{\ell})\in\A[3]$.\qedhere
\end{lemma}

As the largest known common subclass of \WNL\ and of \A[3] is \W[3], the question if $(\BNFAINEshort,\kappa_{\ell})$, belongs to \W[3] is also an interesting \textbf{open question}.

Next, we establish a hardness result for $\co\W[2]$, which even holds for the bounded variant
of the non-emptiness intersection problem. Furthering our previous discussions, notice that $\co\W[2]\subseteq\A[3]$, but little is known about the complement classes of the \W-hierarchy. 

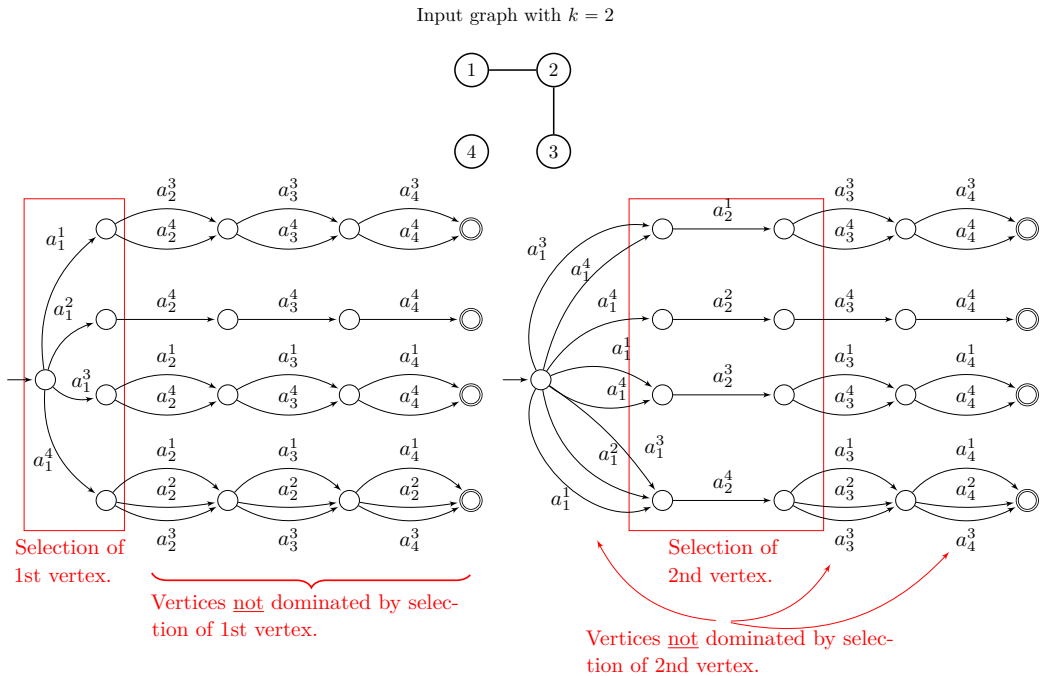
\begin{figure}[htb]
     \centering
    \scalebox{.72}{    
 \begin{tikzpicture}[node distance={15mm}, thick, main/.style = {draw, circle}]  
 \node[main] (1) {$1$};
 \node[main] (2) [right of=1] {$2$};
 \node[main] (3) [below of=2] {$3$};
 \node[main] (4) [left of=3]  {$4$};
 \node (text) at (0.8,1) {Input graph with $k=2$};
 \draw (1) -- (2);
 \draw (2) -- (3);
 \end{tikzpicture}}
 \scalebox{.8}{
 \begin{tikzpicture}[>=latex',shorten >=1pt,node distance=1.5cm and 2cm,on grid,auto]
 \tikzset{every state/.style={minimum size=1pt},initial text={}}
 \node[state,initial] (start) {};
 \node[state] (strand21) [above right= 1cm and 1cm of start] {};
 \node[state] (strand11) [above = of strand21] {};
 \node[state] (strand31) [below right= 0.25cm and 1cm of start] {};
 \node[state] (strand41) [below = 1.75cm of strand31] {};
 
 \node[state] (strand12) [right = of strand11] {};
 \node[state] (strand13) [right = of strand12] {};
 \node[state,accepting] (strand14) [right = of strand13] {};

 \node[state] (strand22) [right = of strand21] {};
 \node[state] (strand23) [right = of strand22] {};
 \node[state,accepting] (strand24) [right = of strand23] {};
 
 \node[state] (strand32) [right = of strand31] {};
 \node[state] (strand33) [right = of strand32] {};
 \node[state,accepting] (strand34) [right = of strand33] {};
 
 \node[state] (strand42) [right = of strand41] {};
 \node[state] (strand43) [right = of strand42] {};
 \node[state,accepting] (strand44) [right = of strand43] {};
 
 \path[->] (start) edge [bend left,pos=.8] node {$a_1^1$}   (strand11)
           (start) edge [bend left,pos=.8] node {$a_1^2$} (strand21)
           (start) edge [bend right,pos=.7] node {$a_1^3$} (strand31)
           (start) edge [bend right,pos=.6,left] node {$a_1^4$} (strand41);

 \path[->] (strand11) edge [bend left]  node {$a_2^3$} (strand12)
           (strand11) edge [bend right] node {$a_2^4$} (strand12);
           
 \path[->] (strand12) edge [bend left]  node {$a_3^3$} (strand13)
           (strand12) edge [bend right] node {$a_3^4$} (strand13);
           
 \path[->] (strand13) edge [bend left]  node {$a_4^3$} (strand14)
           (strand13) edge [bend right] node {$a_4^4$} (strand14);
           
 \path[->] (strand21) edge node {$a_2^4$} (strand22)
           (strand22) edge node {$a_3^4$} (strand23)
           (strand23) edge node {$a_4^4$} (strand24);
           
 \path[->] (strand31) edge [bend left]  node {$a_2^1$} (strand32)
           (strand31) edge [bend right] node {$a_2^4$} (strand32);
 \path[->] (strand32) edge [bend left]  node {$a_3^1$} (strand33)
           (strand32) edge [bend right] node {$a_3^4$} (strand33);
 \path[->] (strand33) edge [bend left]  node {$a_4^1$} (strand34)
           (strand33) edge [bend right] node {$a_4^4$} (strand34);
           
 \path[->] (strand41) edge [bend angle=50,bend left]  node {$a_2^1$} (strand42)
           (strand41) edge [bend angle=10,bend right] node {$a_2^2$} (strand42)
           (strand41) edge [bend right,below] node {$a_2^3$} (strand42);
 \path[->] (strand42) edge [bend angle=50,bend left]  node {$a_3^1$} (strand43)
           (strand42) edge [bend angle=10,bend right] node {$a_2^2$} (strand43)
           (strand42) edge [bend right,below] node {$a_3^3$} (strand43);
 \path[->] (strand43) edge [bend angle=50,bend left]  node {$a_4^1$} (strand44)
           (strand43) edge [bend angle=10,bend right] node {$a_2^2$} (strand44)
           (strand43) edge [bend right,below] node {$a_4^3$} (strand44);
           
 \node[state,initial] (startb) [right = 8.15cm of start] {};
 \node[state] (strand21b) [above right = 1cm and 2cm of startb] {};
 \node[state] (strand11b) [above  = of strand21b] {};
 \node[state] (strand31b) [below right= 0.25cm and 2cm of startb] {};
 \node[state] (strand41b) [below = 1.75cm of strand31b] {};
 
 \node[state] (strand12b) [right = of strand11b] {};
 \node[state] (strand13b) [right = of strand12b] {};
 \node[state,accepting] (strand14b) [right = of strand13b] {};

 \node[state] (strand22b) [right = of strand21b] {};
 \node[state] (strand23b) [right = of strand22b] {};
 \node[state,accepting] (strand24b) [right = of strand23b] {};
 
 \node[state] (strand32b) [right = of strand31b] {};
 \node[state] (strand33b) [right = of strand32b] {};
 \node[state,accepting] (strand34b) [right = of strand33b] {};
 
 \node[state] (strand42b) [right = of strand41b] {};
 \node[state] (strand43b) [right = of strand42b] {};
 \node[state,accepting] (strand44b) [right = of strand43b] {};
 
 \path[->] (strand12b) edge [bend left] node  {$a_3^3$} (strand13b)
           (strand12b) edge [bend right] node {$a_3^4$} (strand13b)
           (strand13b) edge [bend left] node  {$a_4^3$} (strand14b)
           (strand13b) edge [bend right] node {$a_4^4$} (strand14b);
 
 \path[->] (strand22b) edge node {$a_3^4$} (strand23b)
           (strand23b) edge node {$a_4^4$} (strand24b);
 
 \path[->] (strand32b) edge [bend left]  node {$a_3^1$} (strand33b)
           (strand32b) edge [bend right] node {$a_3^4$} (strand33b)
           (strand33b) edge [bend left]  node {$a_4^1$} (strand34b)
           (strand33b) edge [bend right] node {$a_4^4$} (strand34b);
 
 \path[->] (strand42b) edge [bend angle=50,bend left]  node {$a_3^1$} (strand43b)
           (strand42b) edge [bend angle=10,bend right] node {$a_3^2$} (strand43b)
           (strand42b) edge [bend right,below] node {$a_3^3$} (strand43b)
           (strand43b) edge [bend angle=50,bend left]  node {$a_4^1$} (strand44b)
           (strand43b) edge [bend angle=10,bend right] node {$a_4^2$} (strand44b)
           (strand43b) edge [bend right,below] node {$a_4^3$} (strand44b);
 
 \path[->] (startb) edge [bend angle=65, bend left]  node {$a_1^3$} (strand11b)
           (startb) edge [bend angle=25, bend left,above]  node {$a_1^4$} (strand11b);
           
 \path[->] (startb) edge [bend left,pos=.8] node {$a_1^4$} (strand21b);
 
 \path[->] (startb) edge [bend left]  node {$a_1^1$} (strand31b)
           (startb) edge [bend right] node {$a_1^4$} (strand31b);
           
 \path[->] (startb) edge [bend angle=75, bend right, below] node {$a_1^1$} (strand41b)
           (startb) edge [bend angle=35, bend right,pos=.6] node {$a_1^2$} (strand41b)
           (startb) edge [bend angle=10, bend left,pos=.8] node {$a_1^3$} (strand41b);
  
 \path[->] (strand11b) edge node {$a_2^1$} (strand12b)
           (strand21b) edge node {$a_2^2$} (strand22b)
           (strand31b) edge node {$a_2^3$} (strand32b)
           (strand41b) edge node {$a_2^4$} (strand42b);

             \begin{pgfonlayer}{background}
                \draw [draw=red] (-0.35,-2.5) rectangle (1.3,3);
                \draw [draw=red] (9.6,-2.5) rectangle (12.8,3);
                
                
                
                \node[text=red,text width=2cm] at (0.5,-3) {Selection of 1st vertex.};
                \draw [thick,red,decorate,decoration={brace,mirror,amplitude=7pt}] (1.75,-3.2) --  (7,-3.2);
                \node[text=red,text width=5cm] at (4.25,-3.9) {Vertices \underline{not} dominated by selection of 1st vertex.};
                
                \node (start) at (11.15,-4) {};
                \node (t1)    at (9, -2.5) {};
                \node (t2)    at (13, -2.9) {};
                \node (t3)    at (15, -2.7) {};
                
                \path[->] (start) edge [red,bend left] (t1);
                \path[->] (start) edge [red,bend right] (t2);
                \path[->] (start) edge [red,bend right] (t3);
                
                \node[text=red,text width=5cm] at (11.4,-4.5) {Vertices \underline{not} dominated by selection of 2nd vertex.};
                
                \node[text=red,text width=2cm] at (11.25,-3) {Selection of 2nd vertex.};
                \end{pgfonlayer}
 \end{tikzpicture}}
 \caption[NFA Reduction]{Example of the reduction used in the proof of Proposition~\ref{prop:nfa_emptiness_w2}.
  The input graph is written at the top. 
  We have drawn NFA's for the languages given in the proof, where, for better readibility,
  multiple transitions corresponding to a single letter $a_i$ are combined into a single
  transition.
  }
 \label{fig:coW2reduction}
\end{figure}

\begin{proposition}\label{prop:nfa_emptiness_w2}
$\sbounded\BNFAINEshort$, parameterized by $\kappa_{\mathbb{A}}$,
is $\co\W[2]$-hard.
\end{proposition}
\begin{proof}
 We reduce from the well-known $\W[2]$-complete problem {\sc Dominating Set}, 
 parameterized by the size of the solution set~\cite{FluGro2006}.
 Let $G = (V, E)$ be a graph and $k > 0$ be the parameter.
 We want to know if
 there exists a dominating set $D \subseteq Q$, i.e., a set $D$
 such that every vertex in $V$ is either in $D$
 or adjacent to a vertex from $D$, with $|D| \le k$.
 For a vertex $v \in V$, let 
 \( 
 \operatorname{NonDom}(v) = \{ u \in V \setminus \{v\} \mid \{v, u \}\cap E = \emptyset \}
 \)
 be the set of vertices not dominated by $v$ in $G$.
 
 Suppose $V = \{1,2,\ldots, n\}$. Let $\Sigma = \{ a_1, a_2,\ldots, a_n\}$
 be an alphabet. For each $i \in \{1,2,\ldots,k\}$, construct an NFA $\mathcal A_i$
 such that
 \[
  L(\mathcal A_i) = \bigcup_{v \in V} \left[ \left( \prod_{j=1}^{i-1} \{ a_j^u \mid u \in \operatorname{NonDom}(v) \} \right)
   \cdot \{ a_i^v \} \cdot \left( \prod_{j=i+1}^n \{ a_j^u \mid u \in \operatorname{NonDom}(v) \} \right)\right]. 
 \]
  Here, we set $\prod_{i=1}^0 U_i=\{ \varepsilon \}$, and the product operation refers to concatenation extended to subsets of words. 
 The $i$-th NFA $\mathcal A_i$ essentially codifies, on its accepting paths, the $i$-th selection
 of a vertex and which vertices are then not
 dominated by this selection. It is clear
 that we can construct NFA's for these languages with polynomial size in $n$.
 Please see Figure~\ref{fig:coW2reduction}
 for automata corresponding to the languages.
  We have $L(\mathcal A_i) \subseteq \{ a_1, a_1^2, \ldots, a_1^n \}\cdot \{ a_2, a_2^2, \ldots, a_2^n \}\cdots \{ a_n, a_n^2, \ldots, a_n^n \}$
 and, if $a_1^{v_1} a_2^{v_2}\cdots a_n^{v_n} \in \bigcap_{i=1}^k L(\mathcal A_i)$,
 then $v_1, v_2,\ldots, v_k$ correspond to the selection of $k$ vertices
 and $v_{k+1}, v_{k+2},\ldots, v_n$ to vertices not dominated by the previous selection.
 More specifically, for a single NFA~$\mathcal A_i$, $i \in \{1,2,\ldots, k\}$, we have,
 for $u = a_1^{v_1}a_2^{v_2} \cdots a_n^{v_n}$ with $v_i \in V$,
 \[
  u \in L(\mathcal A_i) \Leftrightarrow \mbox{ vertices $\{v_1, v_2,\ldots, v_n \} \setminus \{v_i\}$
   are not dominated by $v_i$. }
 \]
 So,
 $
 u \in \bigcap_{i=1}^k L(\mathcal A_i) \Leftrightarrow 
   \mbox{$\{v_1, v_2,\ldots, v_n \} \setminus \{v_1,v_2, \ldots, v_k\}$
   are not dominated by $\{ v_1, v_2,\ldots, v_k \}$, }
 $
 Hence, we can select $k$ vertices that all dominate the other vertices
 in $G$ if and only if the automata $\mathcal A_i$
 do not accept a common word.
\end{proof}
By the construction used in the previous proof, the hardness results
remains if we add the promise that the input automata accept finite strictly bounded languages.

\section{Conclusion}

In this paper, we studied the parameterized complexity of several problems, mainly concerning variations of the intersection problem of finite automata.
We mainly focused on combinations of natural parameters that are (more than) \W[1]-hard in the case of general automata. For commutative and sparse regular languages, for instance the complexity of $(\DFAINEshort,\kappa_{\mathbb{A}})$ can be brought down from $\XP$ to $\W[1]$. Still, it would be nice (and possibly challenging) to see further practically relevant parameters that show \FPT-membership results.
As there are very few natural problems known to be complete for $\A[3]$, it would be also interesting to see if $(\BNFAINEshort,\kappa_{\mathbb{A}})$ might line up here.
Also the connections to the \TNEJ problem arising in the context of database theory are noteworthy.  In the appendix, we explain in more depth how \TNEJshort is related to the definition of a generalized natural join that admits null values, hopefully reviving the theoretical interest in tables with null values.

The technically most demanding contribution of this paper is \autoref{thm:sparse-DFAINE-W1}. We are currently looking into two possible generalizations of this result:
\begin{itemize}
    \item Do we also get membership in \W[1] for sparse languages (with parameter $\kappa_{\mathbb{A}}$ if we consider nondeterministic finite automata?
    \item Can we consider the number~$m$ of words also as an additional parameter in the \W[1]-membership result?
\end{itemize}
These questions are also open in the strictly-bounded setting, where the parameter~$m$ might appear even more natural. However, due to \autoref{thm:unary-dfa-np-complete}, considering the parameter~$m$ on its own gives a parameterized problem that is not in \XP, unless $\PTIME=\NP$. 

\bibliography{ms} 

\newpage

\section{Appendix: Connections to Database Theory}
We now describe in more details the connection of \TNEJshort to
database theory.
From the mid-70's onwards, there has been quite some discussion on how to
incorporate missing information, also known as null values, into the typical
database operations.

Most relevant to our question is the approach of Lacroix and Pirotte~\cite{DBLP:journals/sigmod/LacroixP76}. They discuss how to deal with null values in connection with the natural join operation.
They describe a concrete example, consisting of Tables $R$ and $S$ with the following tabular representations:

\begin{minipage}{.5\textwidth}
\centering 
\begin{tabular}{c c c}
$R$: &A & B1\\\hline
  &u   & m \\
  &u   & p \\
  &w   &  n\\
  &x   &  p\\
  &x   &  r\\
\end{tabular}
\end{minipage}\begin{minipage}{.5\textwidth}
\centering
\begin{tabular}{c c c}
$S$:&B2 & C\\\hline
  &n   & a \\
  &p   & b \\
  &q   &  c
\end{tabular}
\end{minipage}
How should these tables be joined together with respect to B1 and B2? Obviously, not all B-values are present in either B1 or B2. In order to form the so-called generalized equi-join, denoted by $\textrm{B1}\stackrel{+}{=}\textrm{B2}$ in this example, Lacroix and Pirotte propose to first fill up both B1 and B2 by null values if necessary, which we again denote by $*$, while Lacroix and Pirotte used the symbol~$\omega$. The result could be displayed as follows.
\begin{center}
    \begin{tabular}{c c c c c}
$R[\textrm{B1}\stackrel{+}{=}\textrm{B2}]S$: &A & B1&B2&C\\\hline
  &u   & p & p & b \\
  &w   &  n & n & a\\
  &x   &  p & p & b\\
  &u   & m & $*$ & $*$ \\
  &x   &  r & $*$ & $*$\\
  &$*$ & $*$ & q & c
\end{tabular}
\end{center}
This is the basis to compute the so-called generalized natural join of $R$ and $S$ (with respect to setting $\textrm{B1}=\textrm{B2}$, denoted by $\textrm{B1}\stackrel{+}{*}\textrm{B2}$ in this example,), resulting in:
\begin{center}
    \begin{tabular}{c c c c }
$R[\textrm{B1}\stackrel{+}{*}\textrm{B2}]S$: &A & $\textrm{B1}=\textrm{B2}$&C\\\hline
  &u   & p  & b \\
  &w   & n & a\\
  &x   & p & b\\
  &u   & m  & $*$ \\
  &x   &  r  & $*$\\
  &$*$  & q & c
\end{tabular}
\end{center}
How does this fit together with our treatment?
Let us first generalize our table definition from binary tables to dealing with arbitrary table entry alphabets.

\begin{definition}
A \emph{$\Sigma$-table} (with missing information or null values), or simply a table, is a triple $T = (F, M,\alpha)$ such that $M$ is an $m\times n$ two-dimensional array, $F$ is a one-to-one function such that $\dom(F) = [n]$, and $\alpha:F([n])\to \Sigma$ specifies the symbols that are permitted in the $j$-th column. That is, for all $i\in [m]$ and $j\in [n]$, $M[i,j]\in \alpha(F(j))\cup \{*\}$.  
The rows of~$T$ are exactly the one-dimensional rows of
~$M$.  A \emph{labeled row} of~$T$ is a pair $(F, R)$, where $R$ is a row of~$T$ and $F$ is 
from~$T$.  
\end{definition}

\begin{definition}
Let a $\Sigma_1$-table $T_1 = (F_1, M_1,\alpha_1)$ and a $\Sigma_2$-table $T_2 = (F_2, M_2,\alpha_2)$ be given.  Let a labeled row $x_1 = (F_1, r_1)$ from $T_1$ and a labeled row $x_2 = (F_2, r_2)$ from $T_2$ be given.  We say that $x_1$ and $x_2$ are \emph{compatible} if for all $l \in \range(F_1) \cap \range(F_2)$, $r_1[F_{1}^{-1}(l)] = r_2[F_{2}^{-1}(l)]$ 
or $r_1[F_{1}^{-1}(l)]\notin \alpha_2(l)$ or $r_2[F_{2}^{-1}(l)]\notin \alpha_1(l)$.
\end{definition}

Notice that $r_1[F_{1}^{-1}(l)]\notin \alpha_2(l)$ is also true if $r_1[F_{1}^{-1}(l)] = *$.

\begin{definition}
Let $\mathcal{L} = {(F_i, M_i,\alpha_i)}_{i \in [k]}$ be a family of $k$ tables. 
The \emph{generalized natural join} of $\mathcal{L}$ (denoted by $\Bowtie \, \mathcal{L}$) is a $\Sigma$-table $T = (F, M,\alpha)$ such that:
\begin{itemize}
    \item $range(F) = \bigcup_{i \in [k]} range(F_i)$; $\Sigma=\bigcup_{i \in [k]}\Sigma_i$;
    \item for all $l\in \range(F)$, 
    $\alpha(l)=\bigcup_{i\in [k], l\in \range(F_i)}\alpha_i(l)$;
    \item $M$ consists of \emph{all} rows $R$ (interpreted as $F$-labeled rows) satisfying the following:\\ (1) for all $i \in [k]$, there exists a row $R^{\prime}$ of $M_i$ (interpreted as an $F_i$-labeled row) that is {compatible} with $R$;\\
    (2) for all $l\in \range(F)$, the following two statements are equivalent: (a)  $R[F^{-1}(l)]=*$, and (b) for all $i\in [k]$ with $l\in \range(F_i)$, all rows $R'$ of $M_i$ compatible with $R$ satisfy 
    $R'[F_i^{-1}(l)]= *$. 
\end{itemize}
\end{definition}


From a conceptual point of view, two tables with two (or more) interchanged rows are the same, as we only want to maintain the relational information stored in a table. Similarly, we can also commute columns, as long as the column labeling function $F$ is then also changed accordingly. This defines a natural equivalence relation on the set of tables, and henceforth we will always identify equivalent tables. In actual fact, also our definition of the generalized natural join itself takes this into account, because otherwise the defined operation is not giving a unique result.

The whole situation simplifies a lot if we consider binary tables, i.e., the table alphabet $\Sigma$ equals $\{0,1\}$.
Continuing with \autoref{ex:twobinarytables}, we arrive at:
\\
\begin{minipage}{.3\textwidth}
    \begin{center}
        \begin{tabular}{ c c c c }
         $T_1$: &A & B & C \\
         \hline
         &0 & 0 & $*$ \\
         &1 & $*$ & 1
        \end{tabular}
    \end{center}
\end{minipage}
\begin{minipage}{.3\textwidth}\begin{center}
     \begin{tabular}{ c c c }
         $T_2$: &B & C \\
         \hline
         &0 & 0 \\
         &0 & 1
        \end{tabular}\end{center}
\end{minipage}
\begin{minipage}{.4\textwidth}
\begin{center}
    \begin{tabular}{c c c c }
         $T_1\binBowtie T_2$: &A & B & C \\
         \hline
         &0 & 0 & 0 \\
         &0 & 0 & 1 \\
         &1 & 0 & 1
        \end{tabular}
\end{center}
\end{minipage}

Let us extend this example further by considering the binary table $T_3 := (F_3, A_3)$ where $F_3 = \{ (0, \text{A}), (1, \text{D}) \}$:\\
\begin{minipage}{.3\textwidth}\begin{center}
    \begin{center}
        \begin{tabular}{ c c c }
         $T_3$: &A & D \\
         \hline
         &0 & 1 \\
         &1 & $*$
        \end{tabular}
    \end{center}
\end{center}
\end{minipage}
\begin{minipage}{.3\textwidth}
    \begin{center}
        \begin{tabular}{ c c c c c}
         $T_1\binBowtie T_3$: &A & B & C & D\\
         \hline
         &0 & 0 & $*$ &1\\
         &1 & $*$ & 1& $*$
        \end{tabular}
    \end{center}
\end{minipage}
\begin{minipage}{.4\textwidth}
    \begin{center}
        \begin{tabular}{ c c c c c}
         $T_1\binBowtie T_2\binBowtie T_3$: &A & B & C & D\\
         \hline
         &0 & 0 & 0 & 1 \\
         &0 & 0 & 1 & 1 \\
&1 & 0 & 1 & $*$
        \end{tabular}
    \end{center}
\end{minipage}

As can be seen with the last example, the semantic interpretation of $*$ is that of missing, unconfirmed information. This is different from the interpretation that `all values could be inserted', as then, we might want to obtain the somehow more condensed table 
\begin{center}
        \begin{tabular}{ c c c c c}
         $T$: &A & B & C & D\\
         \hline
         &0 & 0 & $*$ & 1 \\
         &1 & 0 & 1 & $*$
        \end{tabular}
    \end{center}
as a result of joining $T_1$, $T_2$ and $T_3$. However, concerning the null value in the first row $R$ of $T$, we see that there is a row $R'$ in $T_2$ (that also carries the column label C)  that is compatible with $R$ and that satisfies $R'[F_2^{-1}(\textrm{C})]\neq *$; in fact, both rows in $T_2$ may serve as an example.

Let us now return to the example given by Lacroix and Pirotte.
In order to meet our requirements, we identify B1 and B2 in tables $R$ and~$S$ beforehand, calling the corresponding column B. Using $R$ and $S$ as indices, we find that $M_R$ is a $5\times 2$ array, $F_R:[2]\to \{\textrm{A},\textrm{B}\}$, and $\alpha_R:\{\textrm{A},\textrm{B}\}\to \Sigma_R$, with $\Sigma_R=\{\textrm{m},\textrm{n},\textrm{p},\textrm{r},\textrm{u},\textrm{w},\textrm{x}\}$, and that $M_S$ is a $3\times 2$ array, $F_S:[2]\to \{\textrm{B},\textrm{C}\}$, and $\alpha_S:\{\textrm{B},\textrm{C}\}\to \Sigma_S$, with $\Sigma_S=\{\textrm{a},\textrm{b},\textrm{c},\textrm{n},\textrm{p},\textrm{q}\}$. For the general natural join $T=R\binBowtie S$, we obtain
that $M_T$ is a $6\times 3$ array, $F_T:[3]\to \{\textrm{A},\textrm{B},\textrm{C}\}$, and $\alpha_T:\{\textrm{A},\textrm{B},\textrm{C}\}\to \Sigma_T$, with $\Sigma_T=\{\textrm{a},\textrm{b},\textrm{c},\textrm{m},\textrm{n},\textrm{p},\textrm{q},\textrm{r},\textrm{u},\textrm{w},\textrm{x}\}$.
Let us give reasons for the rows shown above for $R[\textrm{B1}\stackrel{+}{*}\textrm{B2}]S$, assuming $F_T(0)=\textrm{A}$, $F_T(1)=\textrm{B}$ and $F_T(2)=\textrm{C}$.
$(\textrm{u},\textrm{p},\textrm{b})$ is a row of $M_T$, as $(\textrm{u},\textrm{p})$ from $R$ is compatible with it, as is $(\textrm{p},\textrm{b})$ from $S$. Similar reasons can be given for $(\textrm{w},\textrm{n},\textrm{a})$ and $(\textrm{x},\textrm{p},\textrm{b})$. As $\textrm{m}\notin \Sigma_S$, the row $(\textrm{u},\textrm{m})$ of $R$ finds no partner in $S$, so that there is no row in $S$ compatible with $(\textrm{u},\textrm{m},*)$, so that the (implicit) implication (contained in item (2) of the definition) `if $R'$ is compatible with $R$, then \dots' is vacuously satisfied.
A similar argument justifies the row $(\textrm{x},\textrm{r},*)$ in~$T$.
Conversely, $\textrm{q}\notin \Sigma_R$ explains row $(*,\textrm{q},\textrm{c})$ in~$T$.

There  are subtle differences between our operation $\Bowtie$ and the operation $\stackrel{+}{*}$ defined in \cite{DBLP:journals/sigmod/LacroixP76}. 
\begin{itemize}
    \item $\Bowtie$ takes an arbitrary number of operands, while $\stackrel{+}{*}$ is a binary operation.
    
    In a sense, this is only a minor issue, although it is the most obvious one. Namely, we could view $\Bowtie$ also as a binary operation. This turns the set~$\mathcal{T}$ of all tables, together with $\Bowtie$, into a groupoid (or magma). It is not hard to see that this operation is associative and commutative. This somehow also justifies our definition that already considers $k$ operands.
    
    \item According to Lacroix and Pirotte, $\stackrel{+}{*}$ should be defined only on two tables (as operands) that do not contain null values. As the resulting table may contain null values, this means that $\stackrel{+}{*}$ cannot define a groupoid. 
\end{itemize}

Let us now turn to a very specific question: When is the generalized natural join of a given family $\cal L$ of tables non-empty, i.e., when does it contain any rows at all? 
According to the definition of~$\Bowtie$, a row $R$ is present in $T=(F,M,\alpha)=\Bowtie\cal L$, if and only if each table $T_i$ in $\cal L$ has a row $R_i$ that is \emph{compatible} with $R$. As the range of $F$ is the union of the ranges of all $F_i$ contained in $T_i=(F_i,M_i,\alpha_i)$, and similarly for the alphabets, this compatibility means that for all $l\in \range(F_i)$, $R_i[F_i^{-1}(l)]=R[F^{-1}(l)]$ or $R[F^{-1}(l)]\notin \alpha_i(l)$ or
$R_i[F_i^{-1}(l)]=*$. 
Now, consider two arbitrary tables $T_i$ and $T_j$ involved in the join.
For all $l\in \range(F_i)\cap \range(F_j)$, we find the following cases, ignoring symmetries:
\begin{enumerate}
    \item $R_i[F_i^{-1}(l)]=R[F^{-1}(l)]=R_j[F_j^{-1}(l)]$,
    \item $R_i[F_i^{-1}(l)]=R[F^{-1}(l)]$ and $R[F^{-1}(l)]\notin \alpha_j(l)$,
    \item $R_i[F_i^{-1}(l)]=R[F^{-1}(l)]$ and $R_j[F_j^{-1}(l)]=*$,
    \item $R[F^{-1}(l)]\notin \alpha_i(l)$ and $R[F^{-1}(l)]\notin \alpha_j(l)$,
    \item $R[F^{-1}(l)]\notin \alpha_i(l)$ and $R_j[F_j^{-1}(l)]=*$.
\end{enumerate}

In the first case, clearly also $R_i$ and $R_j$ are compatible.
In the second case, $R_i[F_i^{-1}(l)]=R[F^{-1}(l)]$ and $R[F^{-1}(l)]\notin \alpha_j(l)$ together imply $R_i[F_i^{-1}(l)]\notin \alpha_j(l)$, ensuring compatibility of $R_i$ and $R_j$. In the third case, if $R_i[F_i^{-1}(l)]=R[F^{-1}(l)]=*$, we find compatibility of  $R_i$ and $R_j$
as in the first case; otherwise, $R_i[F_i^{-1}(l)]=R[F^{-1}(l)]\in\alpha_i(l)$, but $R_j[F_j^{-1}(l)]=*\notin \alpha_i(l)$, i.e., $R_i$ and $R_j$ are compatible.
In the penultimate case, we first consider the case that
$R[F^{-1}(l)]=*$. Then, the definition of~$\Bowtie$ gives further restrictions, implying $R_i[F_i^{-1}(l)]=R_j[F_j^{-1}(l)]=*$ and hence compatibility. 
If $R[F^{-1}(l)]\neq*$, then $R[F^{-1}(l)]\notin \alpha_i(l)$ and $R[F^{-1}(l)]\notin \alpha_j(l)$ implies that (by the definition of the compatibility relation) that both $R_i[F_i^{-1}(l)]$ and $R_j[F_j^{-1}(l)]$ are either undefined or equal $*$. In either case, $R_i[F_i^{-1}(l)]$ and $R_j[F_j^{-1}(l)]$ are compatible.
In the last case, $R_j[F_j^{-1}(l)]=*$ implies $R_j[F_j^{-1}(l)]\notin \alpha_i(l)$, so that we again find that  $R_i$ and $R_j$ are compatible.


\begin{theorem}
Let $\mathcal{L} = {(T_i)}_{i \in [k]}$ be a family of $k$ 
tables, with $T_i=(F_i, M_i,\alpha_i)$. Then, the generalized natural join $\Bowtie\mathcal{L}$ is a table that contains at least one row if and only if there is a family of labeled rows $\{x_i\}_{i \in [k]}$ satisfying the following:
\begin{itemize}
    \item for all $i \in [k]$, $x_i$ is a labeled row of $T_i$;
    \item for all $i$, $j \in [k]$, $x_i$ is compatible with $x_j$.
\end{itemize}
\end{theorem}

\begin{proof}
Our reasoning above shows that if the generalized natural join contains at least one row, then the family of rows $(R_i)$, with $R_i$ picked from $T_i$ as described, is pairwise compatible.

Conversely, consider a selection of labeled rows $x_i$ from $T_i$ that are pairwise compatible. We claim that this means that there is a row in  the generalized natural join $\Bowtie\mathcal{L}$.
This can be seen by induction on $k$.
The result is trivial if $k=1$.
If $k>1$, consider $\mathcal{L}' = {(F_i, M_i,\alpha_i)}_{i \in [k-1]}$. By induction, there is a row $r'$ in $\Bowtie\mathcal{L}'=(F',M',\alpha')$, based on the given collection $(x_i)$ of pairwise compatible labeled rows. Moreover,  $\mathcal{L}' = \mathcal{L}' \binBowtie T_k$. Also, $x_k$ is compatible with $x_i$, for $i\in[k-1]$. We first prove that therefore, $x_k$ is also compatible with $r'$. Consider $l\in \range(F_k)\cap \range(F_i)$, for $i\in[k-1]$. By definition of the generalized natural join, $l\in \range(F_k)\cap \range(F')$. We now distinguish three cases:
\begin{enumerate}
    \item $x_i[F_i^{-1}(l)]=x_k[F_k^{-1}(l)]$;
    \item $x_i[F_i^{-1}(l)]\notin \alpha_k(l)$;
    \item $x_k[F_k^{-1}(l)]\notin \alpha_i(l)$.
\end{enumerate}
In the first case, we can assume $x_i[F_i^{-1}(l)]=x_k[F_k^{-1}(l)]\neq *$, as this case is also covered later. Then, we must have $r'[F'^{-1}(l)]=x_i[F_i^{-1}(l)]=x_k[F_k^{-1}(l)]$, hence $r'$ is compatible with $x_k$.
In the second case, the possibility 
$x_i[F_i^{-1}(l)]= *$ is covered by the last case, as well.
But if $x_i[F_i^{-1}(l)]\neq *$ is defined, then $r'[F'^{-1}(l)]=[F_i^{-1}(l)]$ by definition of the generalized natural join. Hence, $r'[F'^{-1}(l)]\notin \alpha_k(l)$, ensuring compatibility.  
In the last case, by definition of $r'$, $x_k[F_k^{-1}(l)]\notin \alpha'(l)$, yielding compatibility.

Finally, we have to discuss the result of joining $\mathcal{L}'$ with $T_k$. More formally, we discuss the table $\Bowtie\mathcal{L}$, assuming that there is a family of pairwise compatible labeled rows $(x_i)$, with $i\in [k]$. By induction hypothesis, this implies that $\mathcal{L}'$ is not empty. Moreover, $\Bowtie\mathcal{L}=\Bowtie\mathcal{L}'\binBowtie T_k$.
As we have seen, $r'$ (existing by induction hypothesis) is consistent with $x_k$. Now, define $r$ as follows, considering $l\in \range(F)$.
 If $l\in \range(F')$ and $l\in \range(F_k)$, then set:
    $$r[F^{-1}(l)]=\left\{
    \begin{array}{ll}
         x_k[F_k^{-1}(l)],& \text{if } r'[F'^{-1}(l)]=x_k[F_k^{-1}(l)]\\
         x_k[F_k^{-1}(l)],& \text{if }r'[F'^{-1}(l)]=* \text{ and }x_k[F_k^{-1}(l)]\neq *\\
         r'[F'^{-1}(l)],&  \text{if }r'[F'^{-1}(l)]\neq * \text{ and }x_k[F_k^{-1}(l)]= *
    \end{array}\right.$$
    We distinguish several cases according to the definition of compatibility between the labeled rows $r'$ and~$x_k$.
    
    If $x_k[F_k^{-1}(l)]\in \alpha'(l)$ and $r'[F'^{-1}(l)]\in \alpha_k(l)$, then for some $i\in [k-1]$, $x_k[F_k^{-1}(l)]\in \alpha_i(l)$ by definition of $\alpha'$. Hence, by compatibility between $x_i$ and $x_k$, $x_i[F_i^{-1}(l)]=x_k[F_k^{-1}(l)]$ and hence $x_i[F_i^{-1}(l)]=r'[F'^{-1}(l)]$, so that the last subcase of the case distinction in the definition of $r$ sets  $r[F^{-1}(l)]=r'[F'^{-1}(l)]$, guaranteeing the compatibility  of $r$ and~$x_k$.
    
    If $x_k[F_k^{-1}(l)]\in \alpha'(l)$ and $r'[F'^{-1}(l)]\notin \alpha_k(l)$, then we would (correctly) set $r[F^{-1}(l)]=x_k[F_k^{-1}(l)]$ if $r'[F'^{-1}(l)]=*$, ensuring compatibility of $r$ with all $x_i$, $i\in [k]$. But, if  $r'[F'^{-1}(l)]\neq *$,
    then there would be also some $x_i$ with $i\in [k-1]$ such that $x_i[F_i^{-1}(l)]=r'[F'^{-1}(l)]\neq *$, so that $x_i[F_i^{-1}(l)]\in \alpha_i(l)\setminus(\alpha_k\cup\{*\})$.
    This contradicts the assumed compatibility of $x_i$ with $x_k$.
    
    Similarly, if $x_k[F_k^{-1}(l)]\notin \alpha'(l)$ and $r'[F'^{-1}(l)]\in \alpha_k(l)$, then we would (correctly) set $r[F^{-1}(l)]=r'[F'^{-1}(l)]$ if $x_k[F_k^{-1}(l)]=*$, ensuring compatibility of $r$ with all $x_i$, $i\in [k]$. If $x_k[F_k^{-1}(l)]\neq *$, then we again obtain a contradiction to the derived compatibility of $r'$ and~$x_k$.
    
   Finally, assume that $x_k[F_k^{-1}(l)]\notin \alpha'(l)$ and $r'[F'^{-1}(l)]\notin \alpha_k(l)$. Then
    $$(r'[F'^{-1}(l)]=*)\lor (x_k[F_k^{-1}(l)]=*)$$ by construction, because any concrete value in both cases would lead to inconsistencies. The case $(r'[F'^{-1}(l)]=*)\land (x_k[F_k^{-1}(l)]=*)$ is covered by the first case in the case distinction. Then, correctly, $r[F^{-1}(l)]=*$, and $r$ is consistent with $r'$ and $x_k$ and hence, with all $x_i$, $i\in [k]$, (by induction). 
    If $r'[F'^{-1}(l)]=*$ and $x_k[F_k^{-1}(l)]\neq *$, but $x_k[F_k^{-1}(l)]\notin \alpha'(l)$, then we set $r[F^{-1}(l)]=x_k[F_k^{-1}(l)]$ in the second case in the case distinction, so that $r$ is consistent with $r'$ and with $x_k$. The situation $r'[F'^{-1}(l)]\neq *$ and $x_k[F_k^{-1}(l)]= *$ is symmetric.
\end{proof}

By splitting column labels into $\log_2\left(\lceil|\Sigma|\rceil\right)$ many labels, addressing the single bits of non-binary table entries, re-interpreting the natural bijection $[|\Sigma|]\to\Sigma$ on the bit-string level, leading to \emph{binarized variants}, one can see the following results:
\begin{itemize}
    \item The generalized natural join of the binarized variants of the tables $T_i$ equals the binarized variant of the generalized natural join of the tables $T_i$.
    \item The generalized natural join of the tables $T_i$ is non-empty if and only if the generalized natural join of the binarized variants of the tables $T_i$ is non-empty.
    \item The question if the generalized natural join of given tables $T_i$ is non-empty is \NP-complete. Namely, this translates from and into the binary case by the previous observations.
    \item With the parameterization by the number~$k$ of tables, the question of the previous item is \W[1]-complete. 
    
\end{itemize}

So, one can see that the seemingly most simple questions concerning generalized natural joins (with null values) are computationally hard.
This somewhat explains why null values have not been integrated in modern
database query languages. However, then advent of new algorithmic ideas
(as parameterized algorithms) might offer some ways of escaping these affairs.
All that is needed is some effort to return to these areas of database theory. We refer to papers like \cite{Cod79,Zan84} for further attempts to deal with null values in relational databases.

\section{Proving our Main Theorem about Automata Recognizing Sparse Languages}

We need the following classical result, also proven in~\cite{MorRehWel2020},
which, among other compatiblity conditions, will be encoded in an instance
of {\sc Multicolored Clique} in the reduction explained in details below.

\begin{theorem}[Generalized Chinese Remainder Theorem~\cite{schmid58}]
\label{thm:CRT}
 The system of linear congruences
 \(
  x \equiv r_i \pmod{m_i} \quad (i=1,2,\ldots,k)
 \) 
 has integral solutions $x$ if and only if $\gcd(m_i,m_j)$ divides $(r_i - r_j)$
 for all pairs $i \ne j$ and all solutions are congruent modulo $\lcm(m_1, \ldots, m_k)$.
 In short, there exists a natural number $x$
 with $x \equiv r_i \pmod{m_i}$
 for all $i \in\{1,2,\ldots, k\}$
 if and only if $r_i \equiv r_j \pmod{\gcd(m_i, m_j)}$.
\end{theorem}

Let us now formulate and prove our main theorem.

\begin{theorem}
\label{thm:sparse_fixed_words_in_W1}
 Let $w_1, w_2,\ldots, w_m \in \Sigma^*$
 be arbitrary but fixed words. Then, $\DFAINEshort$ for automata accepting languages in $w_1^*w_2^* \cdots w_m^*$, parameterized by $\kappa_{\mathbb{A}}$,
 is $\W[1]$-complete.
\end{theorem}

\begin{proof}\W[1]-hardness follows with Lemma~\ref{lem:CliqueToDFA}, as we can obviously even fix 
$m=|\Sigma|=1$.

 Now, let $w_1, w_2, \ldots, w_m \in \Sigma^*$ be fixed words.
 Suppose we are given $\mathcal A_i = (Q_i, \Sigma, \delta, q_i, F_i)$, $i = 1,2,\ldots,k$,
 accepting languages in $w_1^*w_2^* \cdots w_m^*$,
 with at most $n$ states for each automaton.
 Without loss of generality, each automaton contains only states that are both accessible and co-accessible, possibly except for one not co-accessible dead state. Moreover, we assume that the state sets are pairwise disjoint.
 
 In the following, we assume that all automata
 recognize languages in $w_1^+ w_2^+ \cdots w_m^+$, i.e., every word must be read at least once.  
 This is necessary, because in the construction we put certain edges only in case we can read
 in a word $w_j$, and the construction breaks down if we can ``shortcut'' a block of $w_j$'s
 by not reading a word $w_j$. However, this is actually merely a technicality and could be fixed
 rather easily by solving for all $2^m$ cases that we can leave out a word $w_j$.
 The critical reader might think that this only gives a Turing-reduction, but we can do it in a
 way, by combining several instances of {\sc Multicolored Clique} into a single one,
 that this could indeed be done using only a many-one reduction. We give the details at the end, but first
 describe the core of the reduction under the assumption as stated.

 We are constructing an instance of \textsc{Multicolored Clique}, a well-known \W[1]-complete problem. This means that we are constructing a $k'$-partite graph $G=(V,E)$ that has a clique~$C$ of size $k'$, with one vertex from each class of the partition of the vertex set, if and only if $\bigcap_{i=1}^kL(\mathcal A_i)\neq\emptyset$. As $k'=k\cdot m+1$, this reasoning  proves $(\sparse\DFAINEshort,\kappa_\mathbb{A})\in\W[1]$ for a fixed but arbitrary positive integer~$m$. This is still true after the modification indicated in the previous paragraph, as still the parameter $k'$ of the  resulting \textsc{Multicolored Clique} instance is bounded by a function in $k$ and~$m$.
 
 We first define an auxiliary set $\overline V$. Only a subset of it will be used in the construction of the desired vertex set~$V$.
 More precisely, let
 \[
  \overline V = \{ v_{i, r_1, \ldots, r_m} \mid 
   i \in \{1,2,\ldots,k\} \land 
   r_1, \ldots, r_m \in \{0,1,\ldots,n - 1\} \}
 \] 
 and define mappings $\varphi_i : \overline V \to Q_i$, for $i \in \{1,2,\ldots,k\}$, 
 by
 \[
  \varphi(v_{i, r_1, \ldots, r_m}) 
   = \delta_i(q_i, w_1^{r_1} \cdots w_m^{r_m}). 
 \]
 Notice that $\overline{V}$ is quite big: $|\overline{V}|=k\cdot n^m$. The vertex set $V$ that we are going to define is considerably smaller, but there seems to be no function $f$ such that $|V|\leq f(k,m)$. This prevents us to get a stronger result, with both $k$ and $m$ as parameters for $\W[1]$-membership. 
 In order to state the definition of the vertex set $V$ and later describe the edge set $E$, we need some further notions.
 
 In fact, our strategy is to first capture important aspects of  $\bigcap_{i=1}^kL(\mathcal{A}_i)\neq \emptyset$ by the solvability of a set of equations involving variables $x_j$, $j\in\{1,2,\dots,m\}$, where $x_j$ can take non-negative integer values. We are now describing these equations, derived from $\{\mathcal{A}_i\}_{i=1}^k$.
 
 Let $i \in \{1,2,\ldots,k\}$
 and $r_{i,1}, r_{i,2},\ldots, r_{i,m} \in \{0,1,\ldots,n-1\}$.
 In order to avoid double-indices, we will mostly write $r_j$ instead of $r_{i,j}$ in the following, if the first index is clear from the context. 
 Let $j=j(i,r_1,\dots,r_m) \in \{1,2,\ldots, m\}$
 be the biggest index with $r_j \ne 0$.
 We call $v_i=v_{i,r_1, \ldots, r_m}$
 a \emph{cycle vertex} 
 if there exists an integer $N(v_i) > 0$
 such that 
 $\varphi(v_{i,r_1,\ldots,r_m}) = \delta_i(q_i, w_1^{r_1} \cdots w_{j-1}^{r_{j-1}} w_j^{r_j + N(v_i)})$
 (note that $\varphi(v_{i,r_1,\ldots,r_m}) = \varphi(v_{i,r_1,\ldots,r_j, 0, \ldots, 0}) = \delta_i(q_i, w_1^{r_1} \cdots w_j^{r_j})$), otherwise, 
 we call it a \emph{tail vertex}.
 Choose $N(v_i)=N(i,r_1,\dots,r_m)$ (as described above) minimal in case
 of a cycle vertex.
 Then, associate with the vertex $v_{i, r_1, \ldots, r_m}$ and the numbers $j=j(i,r_1,\dots,r_m)$, $N(v_i)=N(i,r_1,\dots,r_m)$ as above
 the equation 
 \begin{equation}\label{eq:tail}
  x_j = r_{i,j}
 \end{equation}
 if it is a tail vertex, 
 and the equation
 \begin{equation}\label{eq:cycle}
  \left(x_j \equiv r_{i,j} \pmod{N(v_i)}\right) \land x_j \ge r_{i,j}
 \end{equation}
 otherwise, in the case of a cycle vertex. Some explanatory remarks should be in order here:
\begin{itemize}
\item We have written $r_{i,j}$ again in \autoref{eq:tail} and in \autoref{eq:cycle}, because this makes it clearer that (only) these parts of the equations depend on~$i$ and $j$, while the variable $x_j$ does not depend on $i$, but only on $j$.
\item We might have $N(v_i) = 1$ in \autoref{eq:cycle},
 and in this case we add a congruence equation modulo one, which, of course,
 is, strictly speaking, superfluous and could be left out. However, to avoid too many
 case distinctions here and later, we leave those equations in -- at least they do not harm.
\item For notational convenience, partition $\overline{V}_{i,j}$ into the set~$\overline{V}_{i,j}^C$ of cycle vertices and the set~$\overline{V}_{i,j}^T$ of tail vertices.
\end{itemize}

 Note the additional ``threshold'' equation
 $x_j \geq r_j$. For example, if we have a cycle of length two 
 but are required to first read, say, four times the word $w_j$. In this case, the congruence
 equation $x_j \equiv 4\pmod{2}$ alone does not capture this situation (for example $x_j \in \{0,2\}$
 are solutions that should be excluded). However, the additional ``threshold'' equation
 is, strictly speaking, not necessary, as we can add an appropriate multiple
 of all residues to get a bigger solution. So, we can choose a solution
 bigger than any threshold, but we felt that our approach adds explanatory value 
 and makes the construction more transparent.
 
 \begin{claim}
 \label{clm:aut_step_solv_equation}
  Let $i \in \{1,\ldots,k\}$, $j \in \{1,\ldots,m\}$
  and $r_1, \ldots, r_{j-1} \in \{0, 1,\ldots,n-1\}$
  and $r_j > 0$.
  Then, for $s \ge r_j$, we have
  \begin{multline*}
      \delta_i(\varphi(v_{i, r_1, .., r_{j-1},0,..,0}), w_j^s) = \varphi(v_{i,r_1, ..., r_j, 0,...0}) \\ 
      \Leftrightarrow 
      \mbox{$s$ solves the equations associated to $v_{i, r_1, \ldots, r_j, 0, \ldots, 0}$.}
  \end{multline*}
 \end{claim}
 \begin{quote}
     \emph{Proof of the Claim.}
     By the definition of $\varphi$, we have
     $\delta_i(\varphi(v_{i, r_1, \ldots, r_{j-1}, 0, \ldots, 0}), w_j^s) 
      = \delta_i(q_i, w_1^{r_1} \cdots w_{j-1}^{r_{j-1}} w_j^s)$
     and $\varphi(v_{i, r_1, \ldots, r_j, 0, \ldots, 0}) = \delta_i(q_i, w_1^{r_1} \cdots w_j^{r_j})$
     and by assumption both states of $\mathcal A_i$ are equal, i.e.,
     \begin{equation}\label{eqn:automata_equations}
         \delta_i(q_i, w_1^{r_1} \cdots w_{j-1}^{r_{j-1}} w_j^s)
         = \delta_i(q_i, w_1^{r_1} \cdots w_j^{r_j}).
     \end{equation}
      If $v_{i, r_1, \ldots, r_j, 0, \ldots, 0}$
      is a tail vertex, then it is not contained in a cycle for the word $w_j$
      when started from $\delta_i(q_i, w_1^{r_1} \cdots w_{j-1}^{r_{j-1}})$ (recall that by reading
      powers of $w_j$ we essentially trace out a tail and a cycle in $\mathcal A_i$).
       In this case, 
       as $\varphi(v_{i, r_1, \ldots, r_j, 0, \ldots, 0}) = \delta_i(q_i, w_1^{r_1} \cdots w_j^{r_j})$, 
       the equation associated with $v_{i, r_1, .., r_{j-1},0,..,0}$
       is 
       \[
       x_j = r_j.
       \]
      Furthermore, as the state is not contained in a cycle, \autoref{eqn:automata_equations}
      yields $s = r_j$ and so the associated equation is satisfied by $s$.
      Note that in this case, the assumption $s \ge r_j$ was actually superfluous, as it is implied.
      
      Otherwise, $v_{i, r_1, \ldots, r_j, 0, \ldots, 0}$
      is a cycle vertex. Then, choose $N > 0$ minimal (as above) such that 
      $\delta_i(q_i, w_1^{r_1} \cdots w_j^{r_j + N}) = \delta_i(q_i, w_1^{r_1} \cdots w_j^{r_j})$.
      The equation associated with the vertex is
      \[ 
       x_j \equiv r_j \pmod{N} \land x_j \ge r_j.
      \]
      By choice of $N$ and \autoref{eqn:automata_equations}
      we have $s - r_j \equiv 0 \pmod{N}$, as they end up in the same state in
      the cycle induced by powers of $w_j$. By assumption $s \ge r_j$
      and so $s$ solves the associated equation. \emph{[End, Proof of the Claim.]}
 \end{quote}
 
 Observe that in the previous claim, and in terms of the later construction, the condition
 $s \ge r_j$ could actually be removed (compare also the remark about the $x_j \ge r_{i,j}$
 sub-equation in the associated modulo equations). We then only have to modify
 the claim in the sense that then there exists an $s'$
 with $s \equiv s' \pmod{N}$ and $s' \ge r_j$ solving the associated equation. This follows
 as this condition is only relevant for vertices corresponding to cycle states,
 but for them we can enlarge words by the cycle length $N$.

%
 

 \begin{claim} If 
 $\bigcap_{i=1}^kL(\mathcal{A}_i)\neq \emptyset$, then the  set of number-theoretic equations described in \autoref{eq:tail} and in \autoref{eq:cycle}, involving integer variables $x_1,x_2,\dots,x_m$, has a solution $s_1,s_2,\dots,s_m$.
 \end{claim}
 \begin{quote}
 \emph{Proof of the Claim.} 
  Let $w_1^{s_1} w_2^{s_2} \cdots w_m^{s_m} \in \bigcap_{i=1}^kL(\mathcal{A}_i)$.
  For each $i \in \{1,\ldots, m\}$, 
  choose numbers $0 < r_{i, 1}, \ldots, r_{i, m} < n$ (note that here the assumption $L(\mathcal{A}_i) \subseteq w_1^+ w_2^+ \cdots w_m^+$ enters, as it guarantees we can choose positive numbers; 
  also note that multiple possible numbers might be possible to choose from)
  with $\varphi_i(v_{i, r_{i,1}, 0, \ldots, 0}) = \delta_i(q_i, w_1^{s_1})$
  and, inductively, for $1 < j \le m$,
  \[
   \varphi_i(v_{i, r_{i,1}, \ldots, r_{i, j-1}, r_{i,j}, 0, \ldots, 0}) = 
   \delta_i(\varphi_i(v_{i, r_{i,1}, \ldots, r_{i,j-1}, 0, \ldots, 0}), w_j^{s_j}).
  \]
  Note that $\delta_i(\varphi_i(v_{i, r_{i,1}, \ldots, r_{i,j}, 0, \ldots, 0}))
  = \delta_i(q_i, w_1^{s_1} w_2^{s_2} \cdots w_{j}^{s_j})$.
  Note that, by the pigeonhole principle, as the $\mathcal A_i$ have at most $n$
  states, we can find these numbers.
  Then, with \autoref{clm:aut_step_solv_equation},
  the numbers $s_1, \ldots, s_m$,
  solve the associated equations
  for each DFA~$\mathcal A_i$, which follows inductively by the choice of the $r_{i,j}$'s. \emph{[End, Proof of the Claim.]}    
 \end{quote}   
 
 After this number-theoretic interludium, we continue our description of the desired \textsc{Multicolored Clique} instance.
 
 For $i \in \{1,2,\ldots,k\}$
 and $j \in \{1,2,\ldots,m\}$, 
 let \[ \overline V_{i,j} = \{ v_{i, r_1, \ldots, r_m} \in \overline V \mid r_j \ne 0 , r_{j+1} = \ldots = r_m = 0 \}\,.\]
 Hence, using our previous notations, we find that, for $v_{i, r_1, \ldots, r_m}\in \overline V$, $v_{i, r_1, \ldots, r_m}\in \overline V_{i,j}$ if and only if $j=j(i,r_1,\dots,r_m)$.
 
 Actually, $\overline V$
 was just a preliminary set to make the definitions easier, the set
 of actual vertices of the graph we are
 going to construct
 is 
 \begin{equation*}
   V = \{ v_{i, r_1,\ldots, r_m} 
   \in \overline V_{i,j} \mid \exists r_j', \ldots, r_m' \in \{0,1,\ldots,n-1\} : \varphi(v_{i,r_1, \ldots, r_{j-1}, r_j',\ldots, r_m'})\in F_i \} 
   \cup \{ f \}
 \end{equation*}
 The vertices $f$ are necessary to ensure we only select vertices that correspond to final states.
 Then, set $V_{i,j} = \overline V_{i,j} \cap V$, partitioned into the cycle vertices $V_{i,j}^C = \overline V_{i,j}^C \cap V$ and the tail vertices $V_{i,j}^T = \overline V_{i,j}^T \cap V$, 
 and, further, let
  \begin{align*}
     E_1 & = \{ \{ u, v \} \mid u \in  V_{i,j}, v \in V_{i', j'} \mbox{ with } i \ne i' \mbox{ and } j \ne j' \}, \\
     E_2^i & = \{ \{ u, v \} \mid u \in V_{i, j}, v \in V_{i, j'}\mbox{ and } |j - j'| \ge 2 \}, \\ 
     E_2 & = \bigcup_{i=1}^kE_2^i,\\
     E_3^i & = \{ \{ v_{i, r_1, \ldots, r_m}, f \} \mid r_m = 0 \} \cup \{ \{ v_{i, r_1, \ldots, r_m}, f \} \mid r_m > 0 \land \varphi(v_{i,r_1,\ldots, r_m}) \in F_i \}, \\ 
     E_3 & = \bigcup_{i=1}^kE_3^i
 \end{align*}
 and
 \begin{multline*}
     E_4^i = \{ \{v_{i,r_1, \ldots, r_m}, v_{i,r_1', \ldots, r_m'}\} \mid \delta_i(\varphi(v_{i,r_1, \ldots, r_m}), w_{j}^{r_j}) = \varphi(v_{i,r_1', \ldots, r_m'}) \text{ where}\\
     v_{i,r_1, \ldots, r_m} \in V_{i,j-1},
     v_{i,r_1', \ldots, r_m'} \in V_{i,j}, \text{ for some } 2 \le j \le m
     \},
 \end{multline*} 
 so that finally 
 \begin{align*} E_4 & = \bigcup_{i=1}^kE_4^i
 \end{align*}
 
 The edge sets $E_1$
 and $E_2$ are auxiliary connections, introduced so that the vertices picked in the $k'$-partite graph $G=(V,E)$ can actually form a clique.
 The edge sets $E_3$ and $E_4$, respectively,
 are used for 
 the transitions when we switch
 from reading powers of  $w_{j-1}$
 to reading powers of $w_j$, or terminate the reading process if $j=m+1$ (corresponding to $E_3$).
 
 Let $V^i$ collect those vertices from $V$ whose first index equals~$i$, adding the vertex $f$ on top, and let $E^i= E_2^i\cup  E_3^i\cup  E_4^i$.
 We consider the graph $G^i=(V^i,E^i)$ next. Obviously, $G^i$ is $(m+1)$-partite.
 \begin{claim}
 \label{clm:Ai_and_Gi}
 $L(\mathcal{A}_i)\neq\emptyset$ if and only if  $G^i$ contains a multicolored clique (of size $m+1$).
 \end{claim}
 \begin{quote}
  \emph{Proof of the Claim.} First, let $w_1^{s_1} w_2^{s_2} \cdots w_m^{s_m} \in L(\mathcal A_i)$
  with $s_1, \ldots, s_m > 0$.
  We can assume $s_1, \ldots, s_m \le n-1$ (otherwise, a portion of the word induces
  a cycle and could be cut out).
  Select $v_{i, s_1, 0, \ldots, 0}$.
  As $\varphi(v_{i,s_1, \ldots, s_m}) \in F_i$,
  we have $v_{i, s_1, s_2, 0, \ldots, 0} \in V_{i,2}$,
  and, indeed, $v_{i, s_1, \ldots, s_j, 0, \ldots, 0} \in V_{i,j}$
  for every $j \in \{1,2,\ldots,m\}$.
  Further, as $s_2 > 0$, by definition of~$E_4^i$,
  there exists an edge $\{ v_{i, s_1, 0, \ldots, 0}, v_{i, s_1, s_2, 0, 0} \} \in E_4^i$. 
  Similarly, there exists an edge between $v_{i, s_1, s_2, 0, \ldots, 0}$
  and $v_{i, s_1, s_2, s_3, 0, \ldots, 0}$.
  Continuing with this reasoning, there exists an edge
  between $v_{i, s_1, s_2, \ldots, s_{j-1}, 0, \ldots, 0}$
  and $v_{i, s_1, s_2, \ldots, s_{j-1}, s_j, 0, \ldots, 0}$
  for every $j \in \{1,2, \ldots, m\}$.
  By the definition of $E_2^i$, the vertices
  $v_{i, s_1, \ldots, s_j, 0, \ldots, 0}$
  and $v_{i, s_1, \ldots, s_{j'}, 0, \ldots, 0}$
  with $|j - j'| \ge 2$
  are connected by an edge. Finally, by the definition of $E_3^i$, the edge $\{ v_{i, s_1, s_2, s_3, \ldots, s_m},f \}$
  is present. 
  Putting all this together, the set 
  \[
   \{ v_{i, s_1, 0, 0, 0, \ldots, 0},
      v_{i, s_1, s_2, 0, 0, \ldots, 0}, 
      v_{i, s_1, s_2, s_3, 0, \ldots, 0}, 
      \ldots,
      v_{i, s_1, s_2, s_3, \ldots, s_m},
      f \}
  \]
  is a multicolored clique (of size $m + 1$) for $G^i$.

  Conversely, assume we have a multicolored clique of size $m + 1$.
  As we have to select precisely one vertex from each $V_{i,j}$
  and $f$, it has the form
  \[ 
  \{
   v_{i, r_1^{(1)}, 0, \ldots, 0},
   v_{i, r_1^{(2)}, r_2^{(2)}, 0, \ldots, 0},
   \ldots 
   v_{i, r_1^{(m)}, r_2^{(m)}, \ldots, r_m^{(m)}},
   f
  \}.
  \]
  For $V_{i,j}$ and $V_{i,j'}$
  with $|j - j'| \ge 2$, the vertices in $V_{i,j}$
  and $V_{i,j'}$ are always connected by an edge.
  Now, consider
  \[
    v_{i, r_1^{(1)}, 0, \ldots, 0} \text{ and }
    v_{i, r_1^{(2)}, r_2^{(2)}, 0, \ldots, 0}.
  \]
  Both are connected by an edge.
  However, by construction of $G^i$ 
  it must be an edge from $E_4^i$.
  But this implies
  that the word $w_2^{r_2^{(2)}}$
  transfers the state $\varphi(v_{i, r_1^{(1)},0,\ldots,0})$
  into the state $\varphi(v_{i, r_1^{(2)}, r_2^{(2)}, 0,\ldots,0})$
  and that $r_1^{(1)} = r_1^{(2)}$.
  Continuing this reasoning, inductively,
  we find $r_j^{(s)} = r_j^{(m)}$ for $s, j \in \{1,\ldots,m\}$.
  Then, set
  \[
   w = w_1^{r_1^{(m)}} w_2^{r_2^{(m)}} \cdots w_m^{r_m^{(m)}}.
  \]
  With the same reasoning as above (recall $\varphi(v_{i, 0, \ldots, 0}) = q_i$, the start
  state of $\mathcal A_i$),
  \[
   \delta(\varphi(v_{i, 0, \ldots, 0}), w_1^{r_1^{(m)}} \cdots w_j^{r_j^{(m)}}) 
    = \varphi(v_{i, r_1^{(m)}, \ldots, r_j^{(m)}, 0, \ldots, 0}),
  \]
  Now, note that $\varphi(v_{i, r_1^{(m)}, \ldots, r_m^{(m)}}) \in F_i$,
  as we must have an edge $\{ v_{i, r_1^{(m)}, \ldots, r_m^{(m)}}, f \}$,
  which cold only be drawn from $E_3^i$.
  Hence 
  \[ 
  \delta(\varphi(v_{i, 0, \ldots, 0}), w_1^{r_1^{(m)}} \cdots w_m^{r_m^{(m)}}) 
    = \varphi(v_{i, r_1^{(m)}, \ldots, r_m^{(m)}}) \in F_i
  \]
  and so $w_1^{r_1^{(m)}} \cdots w_m^{r_m^{(m)}} \in L(\mathcal A_i)$. \emph{[End, Proof of the Claim.]}
 \end{quote}

 Next, we introduce an edge
 set corresponding to the equations associated with the vertices.
 Recall that we have associated
 with the vertices equations
 over the variables $x_1, x_2,\ldots, x_m$
 that will correspond to the repetition of the words $w_1, w_2,\ldots, w_m$ in a common word.
 Note that $v \in V_{i,j}$ for some $j$
 if and only if $v$
 is associated with an equation
 over the variable $x_j$, as defined in \autoref{eq:tail} and \autoref{eq:cycle}.
 Set
\begin{align*}
 E_5^{C,C} & = \{\, \{ u, v \} \mid  u \in V_{i,j}^C, v \in V_{i',j}^C, i\neq i', r_{i,j} \equiv r_{i',j} \pmod{\gcd(N(u), N(v))}\, \},\\
 E_5^{C,T} & = \{\, \{ u, v \} \mid  u \in V_{i,j}^C, v \in V_{i',j}^T, i\neq i', r_{i,j} \equiv r_{i',j} \pmod{N(u)}\, \},\\  
 E_5^{T,T} & = \{\, \{ u, v \} \mid  u \in V_{i,j}^T, v \in V_{i',j}^T, i\neq i', r_{i,j}= r_{i',j}\, \},\\
 E_5 &= E_5^{C,C}\cup E_5^{C,T}\cup E_5^{T,T}\,.
\end{align*}
This distinguishes all cases of combinations of cycle vertices and tail vertices and their corresponding conditions.
 Then,
 set $G = (V, E)$
 with $E = E_1 \cup E_2 \cup E_3 \cup E_4 \cup E_5$.
 This allows us to formulate a last claim, which follows by combining the arguments of the previous claims.
 
 \begin{claim}The $km + 1$-partite graph 
 $G$
 has a $km + 1$ multicolored
 clique
 (for the independent 
 sets $V_{i,j}$ and $\{ f \}$)
 if and only
 if the automata $\{\mathcal{A}_i\}_{i=1}^k$ accept
 a common word.
 \end{claim}
 \begin{quote}
     \emph{Proof of the Claim.} 
     First, assume the automata accept a common word $w_1^{s_1} w_2^{s_2} \cdots w_m^{s_m}$ (we can
     assume $s_j \in \{0,\ldots, n^k\}$).
     By Claim~\ref{clm:Ai_and_Gi}, every (sub-)graph $G^i$ for $i \in \{1,\ldots,k\}$
     has a multicolored clique of size $m + 1$.
     This has the form (see the proof of Claim~\ref{clm:Ai_and_Gi} for details)
     \[
      \{ v_{i, s_{i,1}, 0, \ldots, 0}, v_{i, s_{i,1}, s_{i,2}, 0, \ldots, 0}, 
        \ldots,
         v_{i, s_{i,1}, \ldots, s_{i,m}}, f\}
     \]
     for numbers $s_{i,j} \in \{0,\ldots,n\}$. We now have to argue
     that the associated equations with both of them are solvable.
     With the same arguments as in Claim~\ref{clm:Ai_and_Gi}, where we can shorten
     portions of the blocks of $w_j$'s when looking at an individual $\mathcal A_i$ by removing cycles,
     we can deduce that
     \[
     \delta(\varphi(v_{i,s_{i,1}, \ldots, s_{i,j-1}, 0, \ldots, 0}), w_j^{s_j}) = v_{i, s_{i,1}, \ldots, s_{i,j}, 0, \ldots, 0}.
     \] 
     This implies that $s_j$ solves the equations associated with 
     $v_{i, s_{i,1}, \ldots, s_{i,j}, 0, \ldots, 0}$, i.e., $s_j = s_{i,j}$
     in case it is a tail vertex and $s_j \equiv s_{i,j} \pmod{N(i, s_{i,1}, \ldots, s_{i,j}, 0, \ldots, 0)}$
     in case it is a cycle vertex.
     As this reasoning was independent of $i$ (and works for every $j$),
     if we have two distinct $i, i' \in \{1,\ldots, k\}$, we can deduce
     that $s_j$
     solves the equations associated with the two vertices
     $
      v_{i, s_{i,1}, \ldots, s_{i,j}, 0, \ldots, 0 } \text{ and }
      v_{i', s_{i',1}, \ldots, s_{i',j}, 0, \ldots, 0}.
     $ 
     By the Generalized Chinese Remainder Theorem, there exists an edge between them.
     Altogether, the set
     \[
      \{f\} \cup \bigcup_{i=1}^k \{ v_{i, s_{i,1}, 0, \ldots, 0}, v_{i, s_{i,1}, s_{i,2}, 0, \ldots, 0}, 
        \ldots,
         v_{i, s_{i,1}, \ldots, s_{i,m}}\}
     \]
     forms a multicolored clique of size $km + 1$ in $G$.

     Now, conversely assume we have a multicolored clique of size $km + 1$.
     Let $i \in \{1,\ldots, k\}$. By only considering 
     the (sub-)graph $G^i$, which then has a multicolored clique
     of size $m + 1$, using Claim~\ref{clm:Ai_and_Gi}
     we can deduce that $L(\mathcal A_i)$
     accepts a word. Inspecting the proof of this
     claim, this word
     has the form
     \[
      w_1^{s_{i,1}} w_2^{s_{i,2}} \cdots w_m^{s_{i,m}}
     \]
     where the vertices of the clique for $G^i$
     are
     \[
      \{ v_{i, s_{i,1}, 0, \ldots, 0}, v_{i, s_{i,1}, s_{i,2}, 0, \ldots, 0}, 
        \ldots,
         v_{i, s_{i,1}, \ldots, s_{i,m}}, f \}
     \]
     Now, for distinct $i, i' \in \{1,\ldots,k\}$
     and $j \in \{1,\ldots,m\}$, the states
     \[
     v_{i, s_{i,1}, \ldots, s_{i,j}, 0, \ldots, 0 } \text{ and }
      v_{i', s_{i',1}, \ldots, s_{i',j}, 0, \ldots, 0}
     \]
     are joined by an edge, which must be from $E_5$, as
     these are the only edges between the independent
     sets $V_{i,j}$ and $V_{i',j}$.
     First, suppose all the vertices for the given
     $j$ and each $i \in \{1,\ldots,k\}$
     are cycle vertices.
     Then, by the Generalized Chinese Remainder Theorem, all these equations
     have a common solution $s_j$.
     If at least one of these vertices
     is a tail vertex, say for the index $i$,
     the equation is $x_j = s_{i,j}$,
     and this must be a solution of all the other equations (congruence, or not) as well.
     In this case, set $s_j$ equal to $s_{i,j}$.
     Then, the word
     \[
      w_1^{s_1} w_2^{s_2} \cdots w_m^{s_m}
     \]
     is accepted by all the automata $\mathcal A_i$.
     This follows, as for $j \in \{1,\ldots,m\}$
     we have, by the way the associated
     equations correspond to the structure
     of the automata, that
     \[
      \delta_i(\varphi(v_{i,s_{i,1}, \ldots, s_{i,j-1}, 0, \ldots, 0}), w_j^{s_j}) 
       = \varphi(v_{i, s_{i,1}, \ldots, s_{i,j}, 0, \ldots, 0})
     \]
     (when $j = 1$, this 
     has to be read as
     $\delta_i(\varphi(v_{i, 0, \ldots, 0}), w_j^{s_j}) 
       = \varphi(v_{i, s_{i,1}, \ldots, s_{i,j}, 0, \ldots, 0})$). 
     So, 
     \[
      \delta_i(q_i, w_1^{s_1} \cdots w_{j}^{s_j})
       = \varphi(v_{i, s_1, \ldots, s_j, 0, \ldots, 0}).
     \]
     Hence, for $j = m$
     and as we have an edge 
     from $v_{i, s_1, \ldots, s_m}$
     to $f$,
     we find
     \[
     \delta_i(q_i, w_1^{s_1} \cdots w_{m}^{s_m}) \in F_i. \text{ \emph{[End, Proof of the Claim.]} }
     \]
     
 \end{quote}
 

 Finally, we have assumed that every word $w_j$ has to appear at least once. Now, we describe
 how to handle the general case.
 For this, we define languages that result if we leave out certain words $w_j$.
 Before giving the formal definition, we give a few examples. If $m = 3$,
 we will define $L^{(1,1,1)} = w_1^*w_1 w_2^* w_2 w_3^* w_3$,
 $L^{(1,0,1)} = w_1^* w_1 w_3^* w_3$, $L^{(1,0,0)} = w_1^* w_1$.
 More formally, let $(i_1, \ldots, i_m) \in \{0,1\}^m \setminus \{(0,\ldots,0)\}$
 and let $L^{(i_1, \ldots, i_m)}$ be the languages such that 
 for every $i_j = 1$, the word $w_j$ has to appear at least once, and if $i_j = 0$
 we skip this block of $w_j$'s.
 We have excluded $L^{(0,\ldots,0)} = \{\varepsilon\}$, as this case could be easily
 checked by the reduction machine itself, by testing for all input automata 
 if the start state is also a final state.
 Then, for each $(i_1, \ldots, i_m)$, we do the above construction
 for the automata $\mathcal A_i'$ with
 \[
  L(\mathcal A_i') = L(\mathcal A_i) \cap L^{(i_1, \ldots, i_m)}.
 \]
 By the product automaton construction, this enlarges the automata $\mathcal A_i'$
 only by a constant factor.
 
 Hence, doing so, we get $2^m - 1$ instances of {\sc Multicolored Clique}.
 However, we can combine them into a single instance of {\sc Multicolored Clique}
 rather easily. For example, if we have two instances $H = (U, F)$
 and $H' = (U', F')$ with independent and disjoint sets $U_1 \cup U_2 \cup \ldots \cup U_r = U$
 and $U_1' \cup U_2' \cup \ldots \cup U_{r'}' = U'$,
 asking if each admits a multicolored clique of size $r$ and $r'$, respectively,
 is the same as asking if in the graph with vertex set $U \cup U'$
 and edge sets
 \[
  E \cup E' \cup \{ \{ u,v \} \mid u \in U, v \in U' \}\,,
 \]
 we have a multicolored clique of size $r + r'$. Repeating this construction
 in our case, we get a single instance $G''`=(V'',E'')$ of {\sc Multicolored Clique}
 with parameter $k''\in\Oh(2^m k (n+1)^m)$.
\end{proof}
For sparse languages, another natural parameter appears, namely the number~$m$ of words that appear in $w_1^*w_2^* \cdots w_m^*$, describing the permitted superset of the language accepted by the considered DFA. This would leads us to parameterized problems like $(\sparse\DFAINEshort,\kappa_{\mathbb{A},m})$.
However, we do not know if $(\sparse\DFAINEshort,\kappa_{\mathbb{A},m})\in\W[1]$. This is an \textbf{open question}.

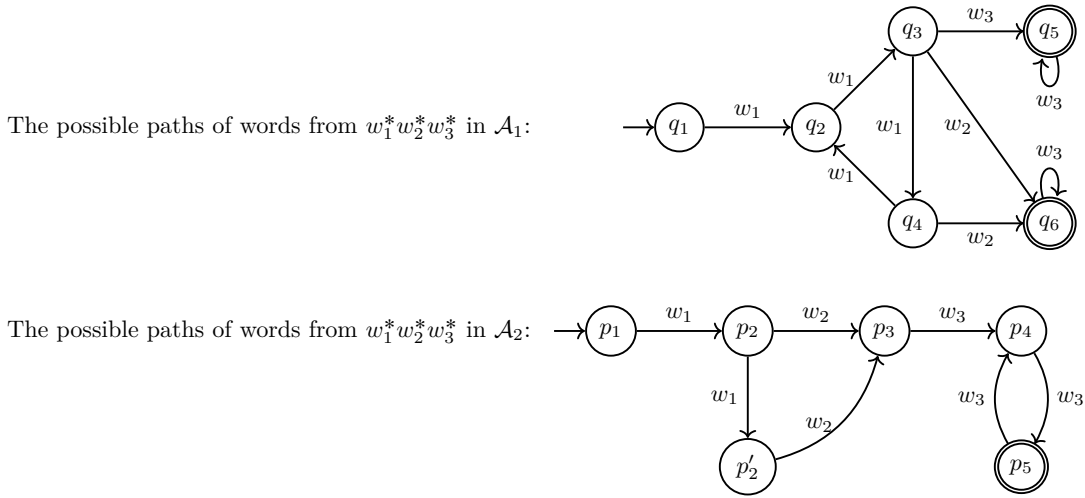
\begin{figure}[htb]
     \centering
    \scalebox{.9}{    
 \begin{tikzpicture}[node distance={20mm}, thick, main/.style = {draw, circle}, initial text={}]  
 \node[main,initial] (1) {$q_1$};
 \node[main] (2) [right of=1] {$q_2$};
 \node[main] (3) [above right of=2] {$q_3$};
 \node[main] (4) [below right of=2] {$q_4$};
 \node[main,accepting] (5) [right of=3] {$q_5$};
 \node[main,accepting] (6) [right of=4] {$q_6$};
 \path[->] (1) edge [above] node {$w_1$} (2);
 \path[->] (2) edge node [left]{$w_1$} (3);
 \path[->] (3) edge node [left]{$w_1$} (4);
 \path[->] (4) edge node [left]{$w_1$} (2);
 \path[->] (3) edge node [above]{$w_3$} (5);
 \path[->] (5) edge [loop below
 ] node {$w_3$} (5);
 \path[->] (3) edge node [left]{$w_2$} (6)
           (4) edge node [below]{$w_2$} (6);
 \path[->] (6) edge [loop above] node {$w_3$} (6);
 
 \node (text) at (-6,0) {The possible paths of words from $w_1^* w_2^* w_3^*$ in $\mathcal A_1$:};
 
 \node[main,initial] at (-1,-3) (p1) {$p_1$};
 \node[main] (p2) [right of=p1] {$p_2$};
 \node[main] (p3) [right of=p2] {$p_3$};
 \node[main] (p4) [right of=p3] {$p_4$};
 \node[main,accepting] (p5) [below of=p4] {$p_5$};
 \node[main] (p2bar) [below of=p2] {$p_2'$};
 
 \node (text) at (-6,-3) {The possible paths of words from $w_1^* w_2^* w_3^*$ in $\mathcal A_2$:};
 
 \path[->] (p1) edge node [above]{$w_1$} (p2)
           (p2) edge node [above]{$w_2$} (p3)
           (p3) edge node [above]{$w_3$} (p4)
           (p4) edge [bend left] node [right]{$w_3$} (p5)
           (p5) edge [bend left] node [left]{$w_3$} (p4);
           
 \path[->] (p2) edge node [left] {$w_1$} (p2bar)
           (p2bar) edge [bend right] node [above, pos=.3] {$w_2$} (p3);
 \end{tikzpicture}}
\caption[Sparse Containment Input Automata]{Example input automata for the reduction described in Theorem~\ref{thm:sparse_fixed_words_in_W1}
 for two input automata $\mathcal A_1$, $\mathcal A_2$, three words $w_1, w_2, w_3$
 and the case $w_1^* w_1 w_2^* w_2 w_3^* w_3$.
 See Figure~\ref{fig:sparse_reduction_full}
 and~\ref{fig:sparse_reduction_simplified}
 for the construction in the proof done for these two input automata.}
 \label{fig:sparse_reduction_input_automata}
\end{figure} 

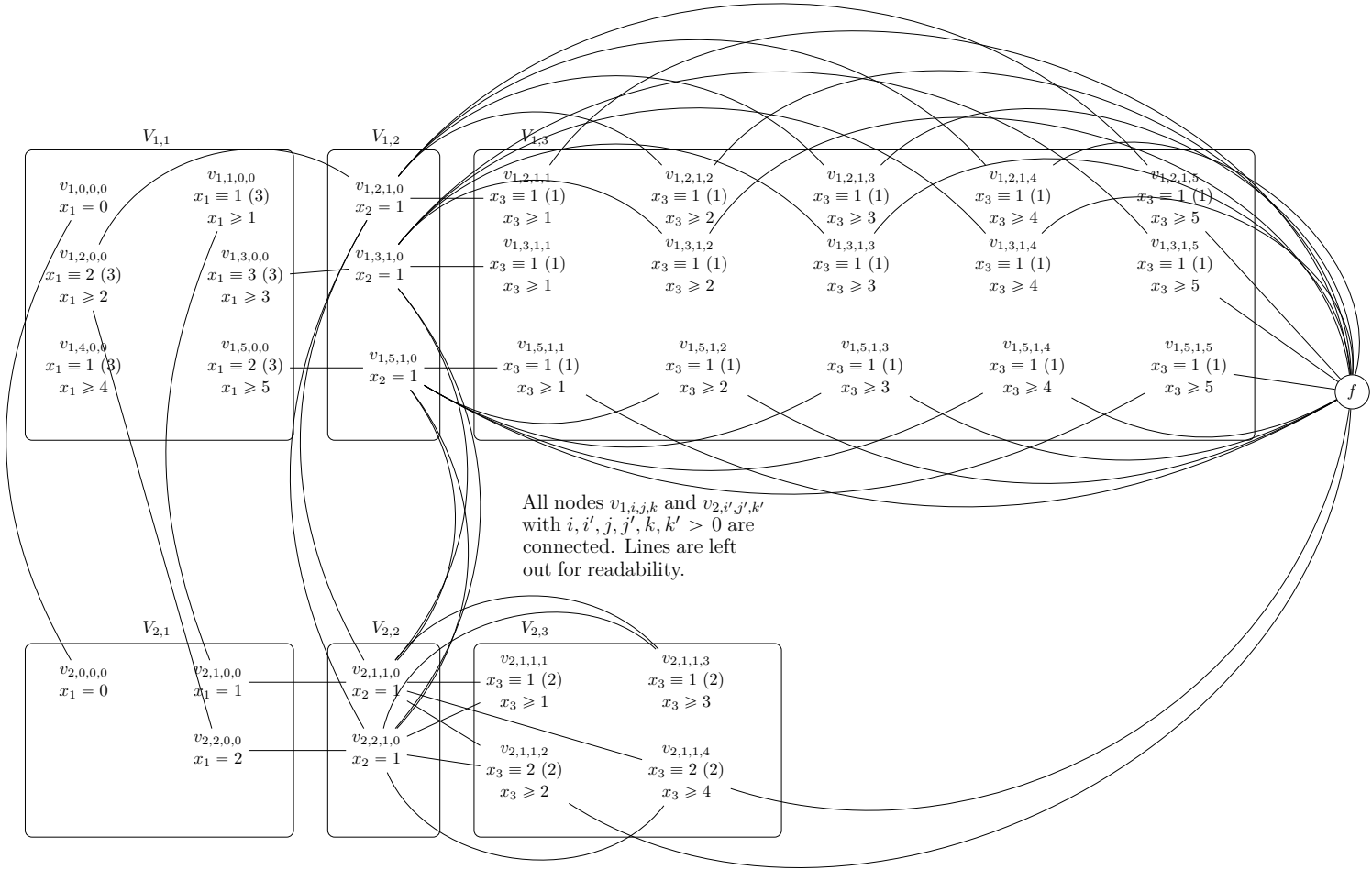
\begin{figure}[htb]
     \centering
  \hspace*{-3cm}
  \scalebox{.7}{\begin{tikzpicture}[node distance=5mm and 15mm, initial text={}] 
    \tikzset{every state/.style={minimum size=1pt}}
    \node [align=center] (v1000) {$v_{1,0,0,0}$ \\ $x_1 = 0$};
    \node [align=center] (v1100) [right = of v1000] {$v_{1,1,0,0}$\\$x_1 \equiv 1 ~ (3)$\\$x_1 \ge 1$};
    \node [align=center] (v1200) [below = of v1000] {$v_{1,2,0,0}$\\$x_1 \equiv 2 ~ (3)$\\$x_1 \ge 2$};
    \node [align=center] (v1300) [right = of v1200] {$v_{1,3,0,0}$\\$x_1 \equiv 3 ~ (3)$\\$x_1 \ge 3$};
    \node [align=center] (v1400) [below = of v1200] {$v_{1,4,0,0}$\\$x_1 \equiv 1 ~ (3)$\\$x_1 \ge 4$};
    \node [align=center] (v1500) [right = of v1400] {$v_{1,5,0,0}$\\$x_1 \equiv 2 ~ (3)$\\$x_1 \ge 5$};

    \node [align=center] (v1210) [right = of v1100] {$v_{1,2,1,0}$\\$x_2 = 1$};
    \node [align=center] (v1310) [below = of v1210] {$v_{1,3,1,0}$\\$x_2 = 1$};
    \node [align=center] (v1510) [right = of v1500] {$v_{1,5,1,0}$\\$x_2 = 1$};

    \node [align=center] (v1211) [right = of v1210] {$v_{1,2,1,1}$\\$x_3 \equiv 1 ~ (1)$\\$x_3 \ge 1$};
    \node [align=center] (v1311) [right = of v1310] {$v_{1,3,1,1}$\\$x_3 \equiv 1 ~ (1)$\\$x_3 \ge 1$};
    \node [align=center] (v1511) [right = of v1510] {$v_{1,5,1,1}$\\$x_3 \equiv 1 ~ (1)$\\$x_3 \ge 1$};

    \node [align=center] (v1212) [right = of v1211] {$v_{1,2,1,2}$\\$x_3 \equiv 1 ~ (1)$\\$x_3 \ge 2$};
    \node [align=center] (v1312) [right = of v1311] {$v_{1,3,1,2}$\\$x_3 \equiv 1 ~ (1)$\\$x_3 \ge 2$};
    \node [align=center] (v1512) [right = of v1511] {$v_{1,5,1,2}$\\$x_3 \equiv 1 ~ (1)$\\$x_3 \ge 2$};
    \node [align=center] (v1213) [right = of v1212] {$v_{1,2,1,3}$\\$x_3 \equiv 1 ~ (1)$\\$x_3 \ge 3$};
    \node [align=center] (v1313) [right = of v1312] {$v_{1,3,1,3}$\\$x_3 \equiv 1 ~ (1)$\\$x_3 \ge 3$};
    \node [align=center] (v1513) [right = of v1512] {$v_{1,5,1,3}$\\$x_3 \equiv 1 ~ (1)$\\$x_3 \ge 3$};
    \node [align=center] (v1214) [right = of v1213] {$v_{1,2,1,4}$\\$x_3 \equiv 1 ~ (1)$\\$x_3 \ge 4$};
    \node [align=center] (v1314) [right = of v1313] {$v_{1,3,1,4}$\\$x_3 \equiv 1 ~ (1)$\\$x_3 \ge 4$};
    \node [align=center] (v1514) [right = of v1513] {$v_{1,5,1,4}$\\$x_3 \equiv 1 ~ (1)$\\$x_3 \ge 4$};
    \node [align=center] (v1215) [right = of v1214] {$v_{1,2,1,5}$\\$x_3 \equiv 1 ~ (1)$\\$x_3 \ge 5$};
    \node [align=center] (v1315) [right = of v1314] {$v_{1,3,1,5}$\\$x_3 \equiv 1 ~ (1)$\\$x_3 \ge 5$};
    \node [align=center] (v1515) [right = of v1514] {$v_{1,5,1,5}$\\$x_3 \equiv 1 ~ (1)$\\$x_3 \ge 5$};

    \node [align=center] (v2000) at (0,-10) {$v_{2,0,0,0}$\\$x_1 = 0$};
    \node [align=center] (v2100) [right = of v2000] {$v_{2,1,0,0}$\\$x_1 = 1$};
    \node [align=center] (v2200) [below = of v2100] {$v_{2,2,0,0}$\\$x_1 = 2$};
    \node [align=center] (v2110) at (6,-10) {$v_{2,1,1,0}$\\$x_2 = 1$};
    \node [align=center] (v2210) [below = of v2110] {$v_{2,2,1,0}$\\$x_2 = 1$};
    \node [align=center] (v2111) [right = of v2110] {$v_{2,1,1,1}$\\$x_3 \equiv 1 ~ (2)$\\$x_3 \ge 1$};
    \node [align=center] (v2112) [below = of v2111] {$v_{2,1,1,2}$\\$x_3 \equiv 2 ~ (2)$\\$x_3 \ge 2$};
    \node [align=center] (v2113) [right = of v2111] {$v_{2,1,1,3}$\\$x_3 \equiv 1 ~ (2)$\\$x_3 \ge 3$};
    \node [align=center] (v2114) [right = of v2112] {$v_{2,1,1,4}$\\$x_3 \equiv 2 ~ (2)$\\$x_3 \ge 4$};

    \node[state] (ensure_only_final_states_selected) at (26,-4) {$f$}; 
    \path (v1211) edge [bend left=70] (ensure_only_final_states_selected);
    \path (v1212) edge [bend left=70] (ensure_only_final_states_selected);
    \path (v1312) edge [bend left=70] (ensure_only_final_states_selected);
    \path (v1213) edge [bend left=70] (ensure_only_final_states_selected);
    \path (v1313) edge [bend left=70] (ensure_only_final_states_selected);
    \path (v1214) edge [bend left=70] (ensure_only_final_states_selected);
    \path (v1314) edge [bend left=70] (ensure_only_final_states_selected);
    \path (v1215) edge (ensure_only_final_states_selected);
    \path (v1315) edge (ensure_only_final_states_selected);
    
    \path (v1511) edge [bend right=30] (ensure_only_final_states_selected);
    \path (v1512) edge [bend right=30] (ensure_only_final_states_selected);
    \path (v1513) edge [bend right=30] (ensure_only_final_states_selected);
    \path (v1514) edge [bend right=30] (ensure_only_final_states_selected);
    \path (v1515) edge (ensure_only_final_states_selected);

    \path (v2112) edge [bend right=60] (ensure_only_final_states_selected);
    \path (v2114) edge [bend right=50] (ensure_only_final_states_selected);
    
    \node at (1.5,1.25) {$V_{1,1}$};
    \draw[rounded corners] (-1.2,1) rectangle (4.3,-5);
     
    \node at (6.2,1.25) {$V_{1,2}$}; 
    \draw[rounded corners] (5,1) rectangle (7.3,-5);
    
    \node at (9.25,1.25) {$V_{1,3}$}; 
    \draw[rounded corners] (8,1) rectangle (24,-5);
    
    \node at (1.5,-8.9) {$V_{2,1}$};
    \draw[rounded corners] (-1.2,-9.2) rectangle (4.3,-13.2);
     
    \node at (6.2,-8.9) {$V_{2,2}$}; 
    \draw[rounded corners] (5,-9.2) rectangle (7.3,-13.2);
    
    \node at (9.25,-8.9) {$V_{2,3}$}; 
    \draw[rounded corners] (8,-9.2) rectangle (14.3,-13.2);
    
    \node[text width=5cm] at (11.5,-7) {\Large All nodes $v_{1,i,j,k}$ and $v_{2,i', j', k'}$ with $i,i',j,j',k,k' > 0$
    are connected. Lines are left out for readability.};

    \path (v1000) edge [bend right=30] (v2000);
    \path (v1100) edge [bend right=22] (v2100);
    \path (v1210) edge [bend right] (v2110);
    \path (v1210) edge [bend right] (v2210);
    
    \path (v1310) edge [bend left=40] (v2110);
    \path (v1310) edge [bend left=40] (v2210);
    
    \path (v1510) edge [bend left=40] (v2110)
          (v1510) edge [bend left=40] (v2210);
    
    \path (v1200) edge [bend left=50] (v1210);
    \path (v1300) edge (v1310);
    \path (v1210) edge (v1211);
    \path (v1310) edge (v1311);
    \path (v2100) edge (v2110);
    \path (v2200) edge (v2210);

    \path (v2110) edge (v2111)
          (v2110) edge (v2112)
          (v2110) edge [bend left=50] (v2113)
          (v2110) edge (v2114);
          
    \path (v2210) edge (v2111)
          (v2210) edge (v2112)
          (v2210) edge [bend left=60] (v2113)
          (v2210) edge [bend right=60]  (v2114);
          
    \path (v1200) edge (v2200);

   \path (v1500) edge (v1510);          
          
    \path (v1210) edge [bend left=50] (v1212);
    \path (v1210) edge [bend left=50] (v1213);
    \path (v1210) edge [bend left=50] (v1214);
    \path (v1210) edge [bend left=50] (v1215);
    
    \path (v1310) edge [bend left=50] (v1312);
    \path (v1310) edge [bend left=50] (v1313);
    \path (v1310) edge [bend left=50] (v1314);
    \path (v1310) edge [bend left=50] (v1315);
    
    \path (v1510) edge (v1511);
    \path (v1510) edge [bend right=30] (v1512);
    \path (v1510) edge [bend right=30] (v1513);
    \path (v1510) edge [bend right=30] (v1514);
    \path (v1510) edge [bend right=30] (v1515);
    
  \end{tikzpicture}}

 \caption[Sparse Containment]{Example of the reduction described in Theorem~\ref{thm:sparse_fixed_words_in_W1}
 for two input automata $\mathcal A_1$, $\mathcal A_2$, three words $w_1, w_2, w_3$
 and the case $w_1^* w_1 w_2^* w_2 w_3^* w_3$. Note that the path that does not use $w_2$
 in $\mathcal A_1$ is not relevant in this case.
 The edges of the sets $E_1$ (edges crossing between $V_{i,j}$'s in different rows
 and columns), $E_2$ (edges between $V_{i,j}$ in the same row, but strictly more than two steps apart, here these are $V_{i,1}$ and $V_{i,3}$) and most of the edges from $E_3$ (edges for the final state selection vertex~$f$) are not drawn for better readability. For the reader wondering about the equations with modulus one: surely they could
 be left out, but as the construction introduces them (see also the remarks in the main text of the proof),
 we also write them here. A common word is $w_1w_1w_2w_3w_3$.
 In Figure~\ref{fig:sparse_reduction_simplified} a simplified drawing is given, as
 here we have drawn in a ```stupid manner'' all vertices from the constructed graph.
 Also note that vertices corresponding to 
 states that are not coaccessible (that is, no final state is reachable)
 by words in $w_1^+ w_2^+ w_3^+$ are not used in the reduction (this is the difference
 between $\overline V$ and $V$ in the proof).}
 \label{fig:sparse_reduction_full}
\end{figure}

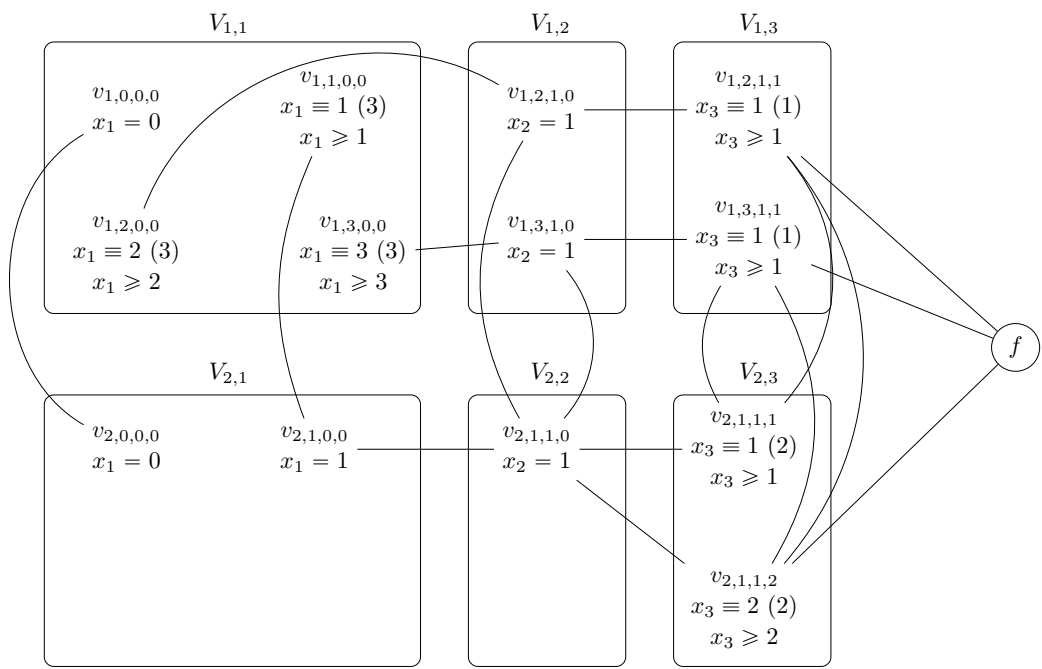
\begin{figure}[htb]
     \centering
   \scalebox{.9}{\begin{tikzpicture}[node distance=10mm and 15mm, initial text={}] 
    \tikzset{every state/.style={minimum size=1pt}}
    \node [align=center] (v1000) {$v_{1,0,0,0}$ \\ $x_1 = 0$};
    \node [align=center] (v1100) [right = of v1000] {$v_{1,1,0,0}$\\$x_1 \equiv 1 ~ (3)$\\$x_1 \ge 1$};
    \node [align=center] (v1200) [below = of v1000] {$v_{1,2,0,0}$\\$x_1 \equiv 2 ~ (3)$\\$x_1 \ge 2$};
    \node [align=center] (v1300) [right = of v1200] {$v_{1,3,0,0}$\\$x_1 \equiv 3 ~ (3)$\\$x_1 \ge 3$};
    
    \node [align=center] (v1210) [right = of v1100] {$v_{1,2,1,0}$\\$x_2 = 1$};
    \node [align=center] (v1310) [below = of v1210] {$v_{1,3,1,0}$\\$x_2 = 1$};
    
    \node [align=center] (v1211) [right = of v1210] {$v_{1,2,1,1}$\\$x_3 \equiv 1 ~ (1)$\\$x_3 \ge 1$};
    \node [align=center] (v1311) [right = of v1310] {$v_{1,3,1,1}$\\$x_3 \equiv 1 ~ (1)$\\$x_3 \ge 1$};
    
    \node [align=center] (v2000) at (0,-5) {$v_{2,0,0,0}$\\$x_1 = 0$};
    \node [align=center] (v2100) [right = of v2000] {$v_{2,1,0,0}$\\$x_1 = 1$};
    \node [align=center] (v2110) at (6,-5) {$v_{2,1,1,0}$\\$x_2 = 1$};
    \node [align=center] (v2111) [right = of v2110] {$v_{2,1,1,1}$\\$x_3 \equiv 1 ~ (2)$\\$x_3 \ge 1$};
    \node [align=center] (v2112) [below = of v2111] {$v_{2,1,1,2}$\\$x_3 \equiv 2 ~ (2)$\\$x_3 \ge 2$};
    
    \node[state] (ensure_only_final_states_selected) at (13,-3.5) {$f$}; 
    \path (v1211) edge (ensure_only_final_states_selected);
    \path (v1311) edge (ensure_only_final_states_selected);
    \path (v2112) edge (ensure_only_final_states_selected);
    
    \node at (1.5,1.25) {$V_{1,1}$};
    \draw[rounded corners] (-1.2,1) rectangle (4.3,-3);
     
    \node at (6.2,1.25) {$V_{1,2}$}; 
    \draw[rounded corners] (5,1) rectangle (7.3,-3);
    
    \node at (9.25,1.25) {$V_{1,3}$}; 
    \draw[rounded corners] (8,1) rectangle (10.3,-3);
    
    \node at (1.5,-3.9) {$V_{2,1}$};
    \draw[rounded corners] (-1.2,-4.2) rectangle (4.3,-8.2);
     
    \node at (6.2,-3.9) {$V_{2,2}$}; 
    \draw[rounded corners] (5,-4.2) rectangle (7.3,-8.2);
    
    \node at (9.25,-3.9) {$V_{2,3}$}; 
    \draw[rounded corners] (8,-4.2) rectangle (10.3,-8.2);
    
    \path (v1000) edge [bend right=60] (v2000);
    \path (v1100) edge [bend right=22] (v2100);
    \path (v1210) edge [bend right] (v2110);
    \path (v1310) edge [bend left=40] (v2110);
    \path (v1311) edge [bend right] (v2111);
    \path (v1200) edge [bend left=50] (v1210);
    \path (v1300) edge (v1310);
    \path (v1210) edge (v1211);
    \path (v1310) edge (v1311);
    \path (v2100) edge (v2110);
    \path (v2110) edge (v2111);
    \path (v2110) edge (v2112);
    \path (v1311) edge [bend left] (v2112);
    
    \path (v1211) edge [bend left=40] (v2111)
          (v1211) edge [bend left=40] (v2112);
  \end{tikzpicture}}
  
 \caption[Sparse Containment Reduced]{Simplified drawing
 of the (full) reduction from Figure~\ref{fig:sparse_reduction_full}. 
 Superfluous vertices are left out.
 The reader might notice that the the left-out vertices in some sense
 copy the structure of the remaining ones and hence
 we have a multicolored clique if and only if we have a multicolored clique in the simplified
 instance. In terms of vertices, the left out vertices correspond
 to states that are already represented, and also their interconnection structure, by the ones we have retained
 in the simplified drawing (note that we cannot simply remove states that are represented multiple times in a \emph{single} automaton, we must consider the maximal number of states in \emph{all} automata). 
 Observe that the reduction in the proof
 does not perform this simplification, and it is only supplied here
 to convey the basic idea better.}
 \label{fig:sparse_reduction_simplified}
\end{figure}

\end{document}